\documentclass[letterpaper]{article} 
\usepackage{aaai23}  
\usepackage{times}  
\usepackage{helvet}  
\usepackage{courier}  
\usepackage[hyphens]{url}  
\usepackage{graphicx} 
\usepackage{natbib}  
\usepackage{caption} 
\usepackage{algorithm}
\usepackage{algorithmic}

\usepackage{newfloat}
\usepackage{listings}
\DeclareCaptionStyle{ruled}{labelfont=normalfont,labelsep=colon,strut=off} 
\floatstyle{ruled}
\newfloat{listing}{tb}{lst}{}
\floatname{listing}{Listing}

\usepackage{graphicx}
\usepackage{xparse}
\usepackage{mathtools}
\usepackage{amsmath}
\usepackage{amssymb}
\usepackage{amsthm}
\usepackage{misc0}
\usepackage{thm-restate}

\theoremstyle{definition}
\newtheorem{theorem}{Theorem}
\newtheorem{example}[theorem]{Example}
\newtheorem{definition}[theorem]{Definition}

\newtheorem{lemma}[theorem]{Lemma}
\newtheorem{proposition}[theorem]{Proposition}
\newtheorem{claim}[theorem]{Claim}

\usepackage[capitalise,noabbrev]{cleveref}

\crefname{subsection}{Subsection}{subsections}
\Crefname{subsection}{Subsection}{Subsections}
\crefname{observation}{Observation}{observations}
\Crefname{observation}{Observation}{Observations}

\usepackage{enumitem}
\usepackage{booktabs}
\usepackage{multirow}

\usetikzlibrary{arrows,backgrounds}

\title{Finite Based Contraction and Expansion via Models}
\author{
    Ricardo Guimarães\textsuperscript{\rm 1},
    Ana Ozaki\textsuperscript{\rm 1},
    Jandson S. Ribeiro\textsuperscript{\rm 2}
}
\affiliations{
    \textsuperscript{\rm 1} University of Bergen\\
    \textsuperscript{\rm 2} University of Hagen\\
    ricardo.guimaraes@uib.no, ana.ozaki@uib.no, jandson.ribeiro@fernuni-hagen.de
}

\begin{document}

\maketitle

\begin{abstract}
We propose a new paradigm for Belief Change in which
the new information is represented as sets of models, while
the agent's body of knowledge is represented as a finite
set of formulae, that is, a finite base.
The focus on finiteness is crucial when we consider limited agents and reasoning algorithms.
Moreover, having the input as arbitrary set of models is more general than the usual treatment of formulae as input.
In this setting, we define new Belief Change operations akin to traditional expansion and contraction, and we identify the rationality postulates that emerge due to the finite representability requirement.
We also analyse different logics concerning compatibility with our framework.
\end{abstract}

\section{Introduction}

The field of Belief Change \citep{Alchourron1985,Hansson1999} studies how an agent should rationally modify its current beliefs when confronted with a new piece of information.
The agent should preserve most of its original beliefs,  minimising loss of information, which is known as the principle of minimal change.
Traditionally,  Belief Change is studied via two perspectives:
(i) set of rationality postulates that conceptualise the principle of minimal change (ii) and classes of Belief Change operations characterised by such rationality postulates.
These two views of Belief Change are tightly connected via representation theorems which show that these views are equivalent.

The standard paradigm of Belief Change~\citep{Alchourron1985},  named AGM due to the initial of its founders,
assumes that an agent's epistemic state is represented as a set of sentences logically closed known as theories.
A main issue with theories is that they are often infinite,  whilst
rational agents are cognitively limited in the sense that an agent is only capable of carrying a finite amount of explicit beliefs,  and all its implicit beliefs follows from such a finite body of beliefs.

A theory that can be generated from a finite set of formulae is called \emph{finite based} \citep{Hansson1996, Hansson93}.
For this reason, we will call finite sets of formulae \emph{finite bases}, as arbitrary sets of formulas simply (belief) \emph{bases}.
While in classical propositional logics, every theory is finite based,  this is  not the case for more expressive logics such as first-order logic (FOL) and several Description Logics (DLs) \citep{baader_horrocks_lutz_sattler_2017} such as \ALC{}.
Computationally,  this finiteness requirement is also important as  reasoners for a particular logic usually can only deal with finite sets of formulae.
Therefore, it is paramount that Belief Change operations guarantee that the new epistemic state is finite based.
However, the AGM postulates do not address finite representation of epistemic states.
In fact, the question of finite representability has not been prioritised in Belief Change.
In this work, we address this issue by devising novel classes of belief operators which ensure that the outcome is finite based.

Moreover,
although the AGM rationality postulates do not depend on any specific logic,  
classes of Belief Change operations have been devised upon strong assumptions about the underlying logics.
In the last years,  effort have been made in replacing some of these assumptions with weaker conditions in order to extend the AGM paradigm 
to more logics such as logics without classical negation \citep{Ribeiro2013},  Horn logics \citep{DelgrandeW13,DelgrandeW10},   temporal logics and logics without compactness \citep{RibeiroNW18,RibeiroNW19AAAI,RibeiroNW19}.
In this work,
we consider that the incoming information is 
represented
as a set of \emph{models},  which generalises the AGM paradigm and other classical Belief Change
frameworks where the incoming information is represented as
formulae in the same logic.
Moreover,
there are scenarios where it is more
convenient that the incoming information
is represented as {models}.
This is the case of the  Learning from Interpretations setting~\citep{DeRaedt1997},
where a formula needs to be created or changed to either incorporate or block
a set of models. \Citet{Arias2007} use this setting to model the construction of Horn theories from graphs.

In logics displaying theories that are not finite based,  the `closest' finite based epistemic state can be chosen instead.
We present an intuitive notion of `closest' finite base to handle cases in which not every theory is finite based.
Using this notion, we then define model change operations which correspond, in
spirit, to expansion and contraction in the AGM paradigm.
We also investigate the rationality consequences of the finiteness requirement and show that our operators only gain or lose models (information) when desired or necessary.
Furthermore, we analyse the compatibility of logics with respect to the emerged rationality postulates, that is, we obtain necessary and sufficient conditions for a logic to admit rational contraction  and expansion operators by models.

In \Cref{sec:prelim}, we briefly review  basic concepts.
In \cref{sec:modelChange}, we detail the new Belief Change paradigm that we propose. 
We discuss in \Cref{sec:uniq}   how properties of a logic (seen as a satisfaction system) affect the behaviour of model change operations.
In \Cref{sec:usecases}, we analyse different logics and the ability to define model change operations following our paradigm.
In \cref{sec:relatedWorks}, we highlight
related works
and conclude in \cref{sec:conc}.
Full proofs of the results can be found in the appendix.

\section{Notation and Basic Notions\label{sec:prelim}}%




The power set of a set $A$ is denoted by $\powerset(A)$,  while the set of all finite subsets of $A$ is denote by $\finitepwset(A)$.
We will write \(\nepwset(A)\) to refer to the non-empty subsets of \(A\).
Following \citet{Aiguier2018} and \citet{DelgrandePW18},  we look at a logic as a satisfaction system.
%
A satisfaction system is a triple \(\logsys=(\llang, \mUni, \models)\),  where
\(\llang\) is a language, \(\mUni\) is the set of models, also
called interpretations, used to give meaning to the sentences in
$\llang$,  and \(\models\) is a satisfaction relation which indicates that a model \(M\)
satisfies a base \(\baseb\) (in symbols, \(M \models
\baseb\)).
%
Looking at a logic simply as a satisfaction system allows us to explore its
properties without making assumptions about the language or putting constraints
upon the logic's entailment relation.  
Our concern is to turn a belief base into a new one that either is satisfied by a given set of models,  or is not satisfied by such models.
Towards this end,  we do not need to constrain how models are used to define
the satisfaction relation,  but rather identify exactly which models satisfy
a belief base $\baseb$ in a satisfaction system \(\logsys = (\llang, \mUni, \models)\) which is given by:
%
%
\[\modelsofx{\baseb}{\logsys}
\coloneqq \{M \in \mathfrak{M} \mid 
M \models
\baseb\}.\]
%
We will write simply \(\modelsof{\baseb}\) when the satisfaction system is clear from the context.
A set of models $\mSet$ within $\logsys$ is finitely representable iff there is
$\baseb \in \finitepwset(\llang)$ such that $\modelsof{\baseb} = \mSet$.  Also,   
we say that a set of formulae $\baseb \subseteq \llang$ is finitely
representable iff there is a $\baseb' \in \finitepwset(\llang)$ with $\modelsof{\baseb} = \modelsof{\baseb'}$.  
The collection of all \emph{finitely representable sets of models in $\logsys$}
is given by:
\[
        \FRsets(\logsys) \coloneqq \{\mSet \subseteq \mUni \mid \exists \baseb \in
        \finitepwset(\llang) : \modelsof{\baseb} = \mSet\}.
    \]

\section{Model Oriented Change on Finite Bases
}\label{sec:modelChange}
In this work,  unlike the standard representation methods in Belief Change,   we consider that:  
incoming information is represented as a (possibly infinite) \textit{set of models}; 
while an agent's epistemic states are represented as \textit{finite (belief) bases}.
Differently, from most approaches in Belief Base Change, we are not concerned 
with
the syntactical structure but, instead, with finiteness.
This notion of belief bases aligns with \citet{Nebel90, dix94} and
\citet{Dalal88},  where a belief base is used simply as a form of finitely representing an agent's epistemic state.
In our setting,  we call each form of rational change in
beliefs 
a \emph{model change operation}. Formally, a model change
operation is a function $f: \finitepwset(\llang) \times \powerset(\mUni) \to \finitepwset(\llang)$.
We propose two kinds of model change operations: \textit{\mexnm{}}
(\(\Rcp(\baseb,\mSet)\)) when we want to accept the input models; and
\textit{\mconnm{}} (\(\Evc(\baseb,\mSet)\)) when we want to reject them instead.
%
%
%
%
%
%

\Mexnm{} turns the current belief state into a new one that is satisfied by the
input models; while in \mconnm{} the new epistemic state is not satisfied by any
of the input models.  In comparison to the Belief Change operations on formulae as
input,  \mexnm{} resembles formula contraction,  as incorporating a new model
implies in removing some formulae from the original belief set.  Analogously,
\mconnm{} resembles formula expansion, as removal of a model implies in
acquisition of information.
In propositional logics,  \mconnm{} and \mexnm{} can be easily defined, as
any set of
models (over a finite signature) is finitely representable. However, in many logics,
there are sets of models that are not finitely representable, even if you
assume that the signature is finite.
We circumvent this issue by adding or removing
models from the current finite base towards the `closest' finite
base satisfied (resp.\ rejected) by the input models.
We show that even with an intuitive notion of `closeness', there
are cases where the `closest' solution does not exist. We also identify
when a solution is uniquely determined.
We introduce each operation separately in the two following
subsections. 

\subsection{\Mconnm{}}%
\label{sec:evc}


The purpose of \mconnm{} is to change the current finite base
$\baseb$ as to forbid any interpretation in the input set $\mSet$.
 If $\modelsof{\baseb} \setminus \mSet$ is not finitely representable, then we
 could simply remove more models until we obtain finite representability.
The question at hand is how many and which models to remove to obtain a finite
representation?  An intuitive idea is to look at a \(\subseteq\)-maximal 
finitely representable subset of $\modelsof{\baseb} \setminus \mSet$.
Such a set is the closest we can get to the ideal result in order to
keep finite representability when subtracting $\mSet$. The class of \mconnm{}
functions we define in this section is based on this idea.  Before we present
them,  let us first introduce some
auxiliary tools.



\begin{definition}%
\label{def:FRsubs}
    Let \(\logsys = (\llang, \mUni, \models)\) be a satisfaction system. Also,
    let \(\mSet \subseteq \mUni\).
    \begin{align*}
        &\FRsubs(\mSet, \logsys) \coloneqq \{ \mSet' \in \FRsets(\logsys) \mid
            \mSet' \subseteq \mSet \\
        &\qquad \text{and} \not\exists \mSet^{''} \in \FRsets(\logsys) \text{ with
    } \mSet' \subset \mSet^{''} \subseteq \mSet \}.
    \end{align*}
\end{definition}

Given a satisfaction system \(\logsys = (\llang, \mUni, \models)\) and a set of models $\mSet \subseteq \mUni$, the set $\FRsubs(\mSet, \logsys)$ contains exactly all the largest (w.r.t.\ set inclusion) finitely representable subsets of $\mSet$.
If we want to contract a set $\mSet$ from a finite base $\baseb$,  then
we can simply build a finite base for one of the sets in $\FRsubs(\modelsof{\baseb}
\setminus \mSet)$.
It turns out that one cannot naively apply this strategy because,  depending on the underlying satisfaction system,  
there might exist a finite base $\baseb$ and set of models $\mSet$ such
that:
\begin{enumerate}[label= (\arabic*), leftmargin=*]
    \item\label{prob1} 
$\FRsubs(\modelsof{\baseb} \setminus \mSet, \logsys) = \emptyset$; or
    \item\label{prob2} $|\FRsubs(\modelsof{\baseb} \setminus \mSet,  \logsys)| \geq 2$.  
\end{enumerate}

If a satisfaction system $\logsys$ displays problem~\ref{prob1} then
we cannot subtract $\mSet$.
{ Thus, we say that a satisfaction system \(\logsys = (\llang, \mUni, \models)\) is \emph{\mconnm-compatible} iff \(\FRsubs(\modelsof{\baseb} \setminus \mSet, \logsys) \neq \emptyset\) for all \(\baseb \in \finitepwset(\llang)\) and \(\mSet \subseteq \mUni\).}
There are two possible causes for problem~\ref{prob1}. First, when a set of models \(\mSet\) has no finitely
representable subset, that is, \(\emptyset \not\in \FRsets(\logsys)\). Second,
when there is no \(\subseteq\)-maximal among infinitely many subsets of
\(\mSet\): for any such subset, there is another subset of \(\mSet\) in
\(\FRsets(\logsys)\) that contains it.
\Cref{fig:4d} illustrates the satisfaction system for propositional Horn logic (\(\logsyshorn\)), a case in which \(\FRsubs\) is always non-empty.
Note that bases in Horn logic can represent only sets of models that are closed under conjunction, which explains why $\{\hmodb, \hmodd\}$ is selected but $\{\hmodb, \hmodc\}$ is not.



\begin{figure}[tb]
    \centering
    \begin{tikzpicture}[scale=1.6]
        \tikzstyle{vertex}=[fill=white, very thin, inner sep=1pt, outer sep=0pt]
        \tikzstyle{selected vertex} = [vertex, draw]
        \tikzstyle{selected edge} = [draw,line width=3pt,->,gray, opacity=0.7]
        \tikzstyle{edge} = [draw,-,black]
        \tikzstyle{pedge} = [draw,->,black, shorten >=2pt]

        \pgfmathsetmacro\sidelength{1.2}
        \pgfmathsetmacro\backx{-0.5}
        \pgfmathsetmacro\backy{0.866}
        \pgfmathsetmacro\transferx{2*\sidelength}
        \pgfmathsetmacro\transfery{0}

        \pgfmathsetmacro\vax{0}
        \pgfmathsetmacro\vay{0}
        \pgfmathsetmacro\vbx{\vax + \sidelength}
        \pgfmathsetmacro\vby{\vay}
        \pgfmathsetmacro\vcx{\vax + \backx}
        \pgfmathsetmacro\vcy{\vay + \backy}
        \pgfmathsetmacro\vdx{\vcx + \sidelength}
        \pgfmathsetmacro\vdy{\vcy}
        \pgfmathsetmacro\vex{\vax}
        \pgfmathsetmacro\vey{\vay + \sidelength}
        \pgfmathsetmacro\vfx{\vbx}
        \pgfmathsetmacro\vfy{\vey}
        \pgfmathsetmacro\vgx{\vcx}
        \pgfmathsetmacro\vgy{\vcy + \sidelength}
        \pgfmathsetmacro\vhx{\vdx}
        \pgfmathsetmacro\vhy{\vgy}
        \pgfmathsetmacro\vix{\vax + \transferx}
        \pgfmathsetmacro\viy{\vay + \transfery}
        \pgfmathsetmacro\vjx{\vix + \sidelength}
        \pgfmathsetmacro\vjy{\viy}
        \pgfmathsetmacro\vkx{\vix + \backx}
        \pgfmathsetmacro\vky{\viy + \backy}
        \pgfmathsetmacro\vlx{\vkx + \sidelength}
        \pgfmathsetmacro\vly{\vky}
        \pgfmathsetmacro\vmx{\vix}
        \pgfmathsetmacro\vmy{\viy + \sidelength}
        \pgfmathsetmacro\vnx{\vjx}
        \pgfmathsetmacro\vny{\vmy}
        \pgfmathsetmacro\vox{\vkx}
        \pgfmathsetmacro\voy{\vky + \sidelength}
        \pgfmathsetmacro\vpx{\vlx}
        \pgfmathsetmacro\vpy{\voy}
        \node[selected vertex] (v0) at (\vax,\vay) {$\emptyset$};
        \node[selected vertex] (v1) at (\vbx,\vby) {$\{\hmodd\}$};
        \node[selected vertex] (v2) at (\vcx,\vcy) {$\{\hmodc\}$};
        \node[selected vertex] (v3) at (\vdx,\vdy) {$\{\hmodc, \hmodd\}$};
        \node[selected vertex] (v4) at (\vex,\vey) {$\{\hmodb\}$};
        \node[selected vertex] (v5) at (\vfx,\vfy) {$\{\hmodb, \hmodd\}$};
        \node[vertex] (v6) at (\vgx,\vgy) {$\{\hmodb, \hmodc\}$};
        \node[vertex] (v7) at (\vhx,\vhy) {$\{\hmodb, \hmodc, \hmodd\}$};
        \node[selected vertex] (v8) at (\vix,\viy) {$\{\hmoda\}$};
        \node[selected vertex] (v9) at (\vjx,\vjy) {$\{\hmoda, \hmodd\}$};
        \node[selected vertex] (v10) at (\vkx,\vky) {$\{\hmoda, \hmodc\}$};
        \node[selected vertex] (v11) at (\vlx,\vly) {$\{\hmoda, \hmodc, \hmodd\}$};
        \node[selected vertex] (v12) at (\vmx,\vmy) {$\{\hmoda, \hmodb\}$};
        \node[selected vertex] (v13) at (\vnx,\vny) {$\{\hmoda, \hmodb, \hmodd\}$};
        \node[selected vertex] (v14) at (\vox,\voy) {$\{\hmoda, \hmodb, \hmodc\}$};
        \node[selected vertex] (v15) at (\vpx,\vpy) {$\mUni$};

        \begin{scope}[on background layer]
            \draw[pedge] (v0) -- (v1);
            \draw[pedge] (v0) -- (v2);
            \draw[pedge] (v0) -- (v4);
            \draw[pedge] (v0.south) to[bend right, looseness=0.35, pedge] (v8.south);

            \draw[pedge] (v1) -- (v3);
            \draw[pedge] (v1) -- (v5);
            \draw[pedge] (v1) to[bend right, looseness=0.35, pedge] (v9);

            \draw[pedge] (v2) -- (v3);
            \draw[pedge] (v2) -- (v6);
            \draw[pedge] (v2) to[bend right, looseness=0.35, pedge] (v10);

            \draw[pedge] (v4) -- (v5);
            \draw[pedge] (v4) -- (v6);
            \draw[pedge] (v4) to[bend left, looseness=0.35, pedge] (v12);

            \draw[pedge] (v8) -- (v9);
            \draw[pedge] (v8) -- (v10);
            \draw[pedge] (v8) -- (v12);

            \draw[pedge] (v3) -- (v7);
            \draw[pedge] (v3) to[bend right, looseness=0.35, pedge] (v11);

            \draw[pedge] (v5) -- (v7);
            \draw[pedge] (v5) to[bend left, looseness=0.35, pedge] (v13);

            \draw[pedge] (v9) -- (v11);
            \draw[pedge] (v9) -- (v13);

            \draw[pedge] (v6) -- (v7);
            \draw[pedge] (v6) to[bend left, looseness=0.35, pedge] (v14);

            \draw[pedge] (v10) -- (v11);
            \draw[pedge] (v10) -- (v14);

            \draw[pedge] (v12) -- (v13);
            \draw[pedge] (v12) -- (v14);

            \draw[pedge] (v7) to[bend left, looseness=0.35, pedge] (v15);
            \draw[pedge] (v11) -- (v15);
            \draw[pedge] (v13) -- (v15);
            \draw[pedge] (v14) -- (v15);
            \draw[selected edge] (v6) to[] (v2);
            \draw[selected edge] (v6) to[] (v4);
            \draw[selected edge] (v7) to[] (v5);
            \draw[selected edge] (v7) to[] (v3);
        \end{scope}
    \end{tikzpicture}
    \caption{\label{fig:4d}Lattice generated by the sets of valuations over the propositional atoms \(\{p, q\}\). Boxed vertices correspond to sets of models in \(\FRsets(\logsyshorn)\). Thin arrows indicate set inclusion, the thick full arrows link sets of models to elements in their respective \(\FRsubs\)}
\end{figure}
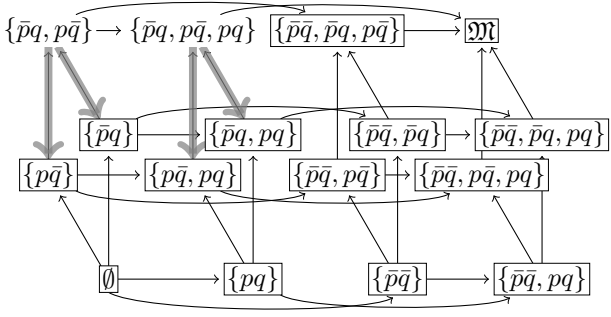

As we will prove in \cref{sec:usecases},
the usual satisfaction systems for propositional logic and propositional Horn
logic are \mconnm-compatible.
However, we will also show that some important satisfaction systems (for instance for
the Description Logic \ALC{}) do not have this property.


There are two alternatives to deal with problem~\ref{prob1}. 
One is to apply 
an approach similar to 
semi-revision \citep{Hansson1997} and reject the change,
keeping the finite base intact.
Another alternative, if \(\emptyset \in \FRsets(\logsys)\), is to impose another
constraint over the plausible candidates.
Problem~\ref{prob2} is related to epistemic choices.
Intuitively $\FRsubs(\modelsof{\baseb}\setminus \mSet, \logsys)$
presents the best solutions to remove $\mSet$. If multiple solutions exist,
then the agent needs to choose among them.
Traditionally, it is assumed that such choices are based on an agent's epistemic
preference over its beliefs,  and such choices are realised by a selection function:   

\begin{definition}%
\label{def:FRsel}
    A \frselnm{} on a satisfaction system \(\logsys\) is a map
    \(\FRsel : \nepwset(\FRsets(\logsys)) \to \FRsets(\logsys)\) such that
%
$\FRsel(X) \in X$.
\end{definition}

Thus, each \frselnm{} determines an \mconnm{} function as follows.

\begin{definition}%
\label{def:FRMcon}
Let \(\logsys\) be an \mconnm-compatible satisfaction system and \(\FRsel\)
a \frselnm{} on $\logsys$. The
\emph{maxichoice \mconnm{} function on \(\logsys\)} defined by \(\FRsel\) is
a map \( \mCon_{\FRsel}: \finitepwset(\llang) \times \powerset(\mUni) \to \finitepwset(\llang)\)
such that: 
    \begin{displaymath}
\modelsof{ \mCon_{\FRsel}(\baseb, \mSet)} = \FRsel(\FRsubs(\modelsof{\baseb}
         \setminus \mSet, \logsys)).
    \end{displaymath}
\end{definition}

The operation $\mCon_{\FRsel}$ chooses exactly one set in
$ \FRsubs$.  \Mconnm{} functions that use this strategy are called
maxichoice 
because by choosing only one element they keep as much
information  as possible from the original finite base.
Another approach  is to allow the selection function to choose multiple
elements,  and then intersect all  of them to build the \mconnm{} result.
However, \cref{propNoPM} shows that
this 
strategy cannot be applied in our setting.

\begin{restatable}{proposition}{noPM}%
\label{propNoPM}
Given a satisfaction system $\logsys = (\llang, \mUni, \models)$ and set of models $\mSet
\subseteq \mUni$,  if $\mathcal{M} \subseteq \FRsubs(\mSet, \logsys)$ and $|\mathcal{M}|\geq 2$
then \textbf{not necessarily} $(\bigcap_{\mSet \in  \mathcal{M}} \mSet) \in\FRsets(\logsys)$. 
\end{restatable}


\Cref{mFRCon_r} states a characterisation of the finitely representable
\mconnm{} functions based on \frselnm{s}. 

\begin{restatable}{theorem}{mFRConRepr}%
\label{mFRCon_r}

    A model change operation  $\mCon$, defined on an \mconnm-compatible
    satisfaction system \(\logsys\),   is a maxichoice \mconnm{} function iff it satisfies the following postulates: 
    \begin{description}[leftmargin=*]
        \item[(success)] \(\mSet \cap \modelsof{\mCon(\baseb, \mSet)} = \emptyset\).
        \item[(inclusion)] \(\modelsof{\mCon(\baseb, \mSet)} \subseteq
            \modelsof{\baseb}\).
		\item[(vacuity)] If $\mSet \cap \modelsof{\baseb} = \emptyset$,  then \\
		\hspace*{1cm} \hfill$\modelsof{\mCon(\baseb, \mSet)} = \modelsof{\baseb}$.
        \item[(finite retainment)] If
            \(\modelsof{\mCon(\baseb, \mSet)}
            \subset \mSet' \subseteq \modelsof{\baseb} \setminus \mSet\) then
            \(\mSet' \not\in \FRsets(\logsys)\).
        \item[(uniformity)] If $\FRsubs(\modelsof{\baseb} \setminus \mSet,
               \logsys) =\FRsubs(\modelsof{\baseb'} \setminus \mSet',
               \logsys)$ then $\modelsof{\mCon(\baseb, \mSet)} = \modelsof{\mCon(\baseb',  \mSet')}$.
    \end{description}
\end{restatable}

The postulate of \emph{success} ensures that no input model will satisfy
the resulting base, while \emph{inclusion} states that no models will
be introduced. \emph{Vacuity} guarantees that models are removed only when the input set has some models in common with the base.
\emph{Finite retainment} expresses the notion of minimality: we
only lose models (other than the input) if there is no other way of ensuring
success, inclusion and vacuity while keeping the base finite. \emph{Uniformity}
states that the result depends only on \(\FRsubs\).
\textit{Vacuity} is redundant in the presence of \textit{inclusion} and \textit{finite retainment}.
\begin{restatable}{proposition}{vacuity}
If a model change operation $\mCon$ satisfies \textit{inclusion} and \textit{finite retainment},  then it satisfies \textit{vacuity}.
\end{restatable}
An analogous to the classical recovery postulate would be desirable:
if a set of models $\mSet$ is evicted from a finite base $\baseb$,  then putting $\mSet$ back should restore all the models of $\baseb$.
This `model-recovery' postulate, however,  cannot be satisfied: 
in order to evict $\mSet$,  some extra models might be purged in order to reach a finite base,  and they cannot be restored by simply putting $\mSet$ back.
%
Although the roles of the postulates \textit{conjunction} and \textit{intersection} are well-known within classical logics, understanding their behaviours within non-classical settings are still a challenge \citep{RibeiroNW18, RibeiroNW19AAAI}.
While \textit{intersection} follows directly from finite-retainment, we cannot characterise \textit{conjunction} since our framework goes beyond the classical case.  

\subsection{\Mexnm{}}

\Mexnm{} alters a finite base \(\baseb\) to incorporate all models in
$\mSet$.
In some satisfaction systems, $\modelsof{\baseb} \cup \mSet$ is not finitely representable.
Analogous to the strategy employed in the previous \lcnamecref{sec:evc},
\mexnm{} can be defined using the smallest supersets of
$\modelsof{\baseb} \cup \mSet$.


\begin{definition}%
\label{def:FRsups}
    Let \(\logsys = (\llang, \mUni, \models)\) be a satisfaction system.
    Also,   let \(\mSet \subseteq \mUni\).
    \begin{align*}
        &\FRsups(\mSet, \logsys) \coloneqq \{ \mSet' \in \FRsets(\logsys) \mid
            \mSet \subseteq \mSet' \\
        &\qquad \text{and} \not\exists \mSet^{''} \in \FRsets(\logsys) \text{ with
    } \mSet \subseteq \mSet^{''} \subset \mSet'\}.
    \end{align*}
\end{definition}

There are also satisfaction systems \(\logsys = (\llang, \mUni,
\models)\) with \(\FRsups(\mSet, \logsys) = \emptyset\) for some \(\mSet \subseteq \mUni\).
The causes are dual to the \mconnm{}
case: either \(\mUni \not\in \FRsets(\logsys)\) or there is a \(\mSet \subseteq
\mUni\) without a \(\subseteq\)-minimal superset in \(\FRsets(\logsys)\). 
{We say that a satisfaction system \(\logsys = (\llang, \mUni, \models)\) is \emph{\mexnm-compatible} iff \(\FRsups(\modelsof{\baseb} \cup \mSet, \logsys) \neq \emptyset\) for all \(\baseb \in \finitepwset(\llang)\) and \(\mSet \subseteq \mUni\).}
\Cref{fig:4d} also shows a situation in which the satisfaction system is \mexnm-compatible.
In such systems, we can design \mexnm{} as follows.

\begin{definition}%
\label{def:FRMexp}
    Let \(\logsys = (\llang, \mUni, \models)\) be a \mexnm-compatible satisfaction system and
    \(\FRsel\) a \frselnm{} on \(\logsys\).
    The \emph{maxichoice model
    \mexnm{} function on \(\logsys\)} defined by \(\FRsel\) is
    a map \( \mExp_{\FRsel}: \finitepwset(\llang) \times \powerset(\mUni) \to \finitepwset(\llang)\)
    such that:
    \begin{displaymath}
         \modelsof{\mExp_\FRsel(\baseb, \mSet)} = \FRsel(\FRsups(\modelsof{\baseb} \cup
         \mSet, \logsys)).
    \end{displaymath}
\end{definition}

An analogous of \cref{propNoPM} also holds for \mexnm{}, as stated in \cref{propNoPMR}.

\begin{restatable}{proposition}{noPMR}%
\label{propNoPMR}
Given a satisfaction system $\logsys = (\llang, \mUni, \models)$ and set of models $\mSet
\subseteq \mUni$,  if $\mathcal{M} \subseteq \FRsups(\mSet, \logsys)$ and $|\mathcal{M}|\geq 2$
then \textbf{not necessarily} $(\bigcup_{\mSet \in  \mathcal{M}} \mSet) \in\FRsets(\logsys)$. 
\end{restatable}

In \cref{sec:usecases}, we will show that
the usual satisfaction systems for propositional logic and proposition Horn
logic are also \mexnm-compatible.
We will also introduce a
satisfaction system that is \mexnm-compatible but not \mconnm, thus, showing that
\mexnm-compatibility and
\mconnm-compatibility are not always co-occurrent. A satisfaction system
\(\logsys = (\llang, \mUni, \models)\) can be such that \(\emptyset \in \FRsets(\logsys)\) but \(\mUni
\not\in \FRsets(\logsys)\), and vice-versa.
We identify the 
set of rationality postulates that characterise the \mexnm{} function from \cref{def:FRMexp}.

\begin{restatable}{theorem}{mFRExpRepr}%
\label{mFRExp_r}
A model change operation \(\mExp\), defined on a \mexnm-compatible satisfaction
system \(\logsys\),
    is a maxichoice \mexnm{} function 
    iff it satisfies the following
    postulates:
    \begin{description}[leftmargin=*]
        \item[(success)] \(\mSet \subseteq \modelsof{\mExp(\baseb, \mSet)}\).
        \item[(persistence)] \(\modelsof{\baseb} \subseteq
            \modelsof{\mExp(\baseb, \mSet)} \).
        \item[(vacuity)]  $\modelsof{\mExp(\baseb,\mSet)}
            = \modelsof{\baseb}$, if $\mSet \subseteq \modelsof\baseb$.
        \item[(finite temperance)] If
            \(\modelsof{\baseb} \cup \mSet \subseteq \mSet' \subset
            \modelsof{\mExp(\baseb, \mSet)}\) then \(\mSet'
            \not\in \FRsets(\logsys)\).
            \item[(uniformity)] If $\FRsups(\modelsof{\baseb} \cup \mSet,  \logsys)
                =\FRsups(\modelsof{\baseb'} \cup \mSet',  \logsys)$ then $\modelsof{\mExp(\baseb, \mSet)} =\modelsof{ \mExp(\baseb',
                \mSet')}$.
    \end{description}
\end{restatable}

The postulates presented in \cref{mFRExp_r} are straightforward
translations of the classical framework of Belief Change expansion, being
finite temperance the only which deviates w.r.t.\ its classical correspondent.
\emph{Success} guarantees that the input models will satisfy
the resulting base, while \emph{persistence} determines that no model will
be lost.  \textit{Vacuity} ensures that models will be added only when the input set brings new models.
\emph{Finite temperance} expresses the notion of minimality: we
only gain models (other than the input) if there is no other way of ensuring
success and persistence while keeping the base finitely representable.
\emph{Uniformity} states that the result depends only on \(\FRsups\).
\textit{Vacuity} is redundant in the presence of \textit{finite temperance} and \textit{persistence}.
\begin{restatable}{proposition}{vacuityExp}
If a model change operation $\mExp$ satisfies \textit{persistence} and \textit{finite temperance}, then it satisfies \textit{vacuity}.
\end{restatable}

We can also translate the postulate \textit{monotony} from classical expansion
to our setting as follows: if $\modelsof{\baseb} \subseteq \modelsof{\baseb'}$ then \(\modelsof{\mExp(\baseb,
\mSet)} \subseteq \modelsof{\mExp(\baseb', \mSet)}\). However, \(\mExp\) does not satisfy this
postulate and enforcing it means imposing monotonicity on the
operation $\mExp$ similar to what happens to the \textit{update} operations of
\citet{Katsuno1991}. We would have to constrain \frselnm{} to only pick certain elements of
\(\FRsups\).
A third operation of Belief Change on formulae is belief revision whose purpose is to incorporate a new piece of information and guarantee that the new theory is consistent.
In terms of models as input,  we could define the model revision operation
whose purpose would be to remove models but avoiding that the inconsistent state is reached.
	To avoid the inconsistent state,  the agent would need to select a `closest' finitely representable set of models according to its underlying epistemic preference relation.
	We leave such investigation as future work.

\section{\label{sec:uniq}Uniqueness and Characterisation}


	In some satisfaction systems, the result of any \mconnm{} is uniquely
	determined by the input models and initial base, regardless of the underlying
	\frselnm{}. The same holds for \mexnm{} in some systems. Many well-known satisfaction systems such as the
	traditional ones for propositional logic and propositional Horn logic have the \emph{reverse
		monotonic bijection property (RMBP)}. 


\begin{definition}%
	\label{def:RBMP}
		A  satisfaction system \(\logsys = (\llang, \mUni, \models)\)
		has the 
			RMBP 
			if
			for every \(\baseb_1, \baseb_2 \subseteq \llang\) and every \(\modelm \in
			\mUni\): \(\modelm \in \modelsof{\baseb_1}\) and  \(\modelm \in
			\modelsof{\baseb_2}\) iff \(\modelm \in \modelsof{\baseb_1 \cup
				\baseb_2}\).
	\end{definition}

	\Cref{uniqueFR} states the RMBP is a sufficient condition for this determinism.
	\begin{restatable}{proposition}{UniqueFR}\label{uniqueFR}
		Let \(\logsys = (\llang, \mUni, \models)\) be a satisfaction system with the
		RMBP.\@ Then \(|\FRsups(\mSet, \logsys)| \leq 1\) and \(|\FRsubs(\mSet,
		\logsys)| \leq 1\) for all \(\mSet \subseteq \mUni\).
	\end{restatable}

		Due to \cref{uniqueFR}, if \(\logsys = (\llang, \mUni, \models)\) has the RBMP
		then every \frselnm{} will yield the same result when applied over
		\(\FRsups(\mSet, \logsys)\) for any \(\mSet \subseteq \mUni\), and the same
		holds for \(\FRsubs(\mSet, \logsys)\). 


We devote the rest of this section to prove
a characterization of eviction- and reception-compatibility
based on the notion of partial orders.
The intuitive idea is that \mconnm-compatibility of a satisfaction system \(\logsys\) depends on the ability of finding at least one subset which can be seen as the `immediate predecessor' when adding a set of models to the partially ordered set (poset) \((\FRsets(\logsys), \subset)\).

\begin{definition}%
	\label{def:immediate}
	Let \((P, \preceq)\), \(x, y, z \in P\) and \(\prec\) the strict version of \(\preceq\).
    We say that \(x\) is an \emph{immediate predecessor of \(y\)} if \(x \prec y\) and there is no \(x^\prime \in P\) with \(x \prec x^\prime \prec y\).
    Analogously, we say that \(z\) is an \emph{immediate successor of \(y\)} if \(y \prec z\) and there is no \(z^\prime \in P\) with \(y \prec z^\prime \prec z\).
\end{definition}

If a satisfaction system \(\logsys = (\llang, \mUni, \models)\) guarantees that for any \(\baseb \in \finitepwset(\baseb)\) and \(\mSet \subseteq \mUni\), \(\modelsof{\baseb} \setminus \mSet\) will have a finitely representable immediate predecessor regarding set inclusion (\(\subseteq\)), then it is \mconnm-compatible.
Some satisfaction systems, do not guarantee this because the empty set of models is not representable (there is no inconsistent base).
That would be case for propositional Horn logic if we removed the constant \(\perp\).
On the other hand, some satisfaction system (as we will see in \cref{sec:usecases}) have non-finitely representable sets of models for which there are arbitrarily close approximations.
Hence, none of the infinitely many candidates is an immediate predecessor w.r.t.\ set inclusion.
The analogous notions and observations hold for \mexnm-compatibility.
\Cref{ex:rational} illustrates one such satisfaction system.

\begin{example}%
\label{ex:rational}
    Let \(\logsys_{q} = (\llang_q, \mUni_q, \models_q)\) be such that
    \(\llang_q = \{[x, y] \mid x, y \in \mathbb{Q} \text{ and } x \leq y\}\),
    \(\mUni_q = \mathbb{Q}\) and
    \(Q \models_q \baseb\) (with \(Q \subseteq \mathbb{Q}\)) iff for all \(z \in Q\), \(x \leq z \leq y\) for every \([x, y] \in \baseb\).

    Intuitively, every finite base either has no models, or corresponds to a closed interval on the rationals.
    However, the target set of models produced by an eviction or reception can correspond to an open interval.
    For eviction, take the base \(\{[0, 1]\}\) and the set of models \(\{1\}\) and for reception, take the base \(\{\mathopen[0.5,1\mathclose]\}\) and the set of models \(\mathopen(0, 1\mathclose]\).
    In both cases, one can find arbitrarily close approximations, thus there might be no maximal subset for eviction nor a minimal superset for reception.
\end{example}

Given a satisfaction system \(\logsys = (\llang, \mUni, \models)\), it is not only the density of \((\FRsets(\logsys), \subset)\) that determines compatibility.
Even when the poset is dense, if every set of models is finitely representable (that is, \(\FRsets(\Lambda) = 2^\mUni\)) then \(\logsys\) is clearly \mconnm- and \mexnm-compatible.
Using \cref{def:immediate} we can finally characterise \mconnm- and \mexnm-compatibility with the following
theorem.

\begin{restatable}{theorem}{IncFRsubs}%
	\label{incFRsubs}
	A satisfaction system \( \logsys = (\llang, \mUni, \models)\)
    is
    \begin{itemize}
    	\item  \mconnm-compatible iff for every \(\mSet \subseteq \mUni\) either (i) \(\mSet \in \FRsets(\logsys)\), (ii) \(\mSet\) has an immediate predecessor in
    	\((\FRsets(\logsys) \cup \{\mSet\}, \subset)\), or (iii) there is no \(\mSet' \in \FRsets(\Lambda)\) with \(\mSet \subseteq \mSet'\); and
    	\item \mexnm-compatible iff for every \(\mSet \subseteq \mUni\) either (i) \(\mSet \in \FRsets(\logsys)\), (ii) \(\mSet\) has an immediate successor in
    	\((\FRsets(\logsys) \cup \{\mSet\}, \subset)\), or (iii) there is no \(\mSet' \in \FRsets(\Lambda)\) with \(\mSet' \subseteq \mSet\).
    \end{itemize}

\end{restatable}


While verifying compatibility can be very cumbersome in general, \cref{evcRcpSuff} displays a simpler sufficient condition when \(\FRsets(\logsys)\) is finite.

\begin{restatable}{corollary}{EvcRcpSuff}%
	\label{evcRcpSuff}
	Let \(\logsys = (\llang, \mUni, \models)\) be satisfaction system in which
	\(\FRsets(\logsys)\) is finite. Then:
	\begin{itemize}
		\item \(\logsys\) is \mconnm-compatible iff \(\emptyset \in \FRsets(\logsys)\).
		\item \(\logsys\) is \mexnm-compatible iff \(\mUni \in \FRsets(\logsys)\).
	\end{itemize}
\end{restatable}

\section{Compatibility: Use Cases}\label{sec:usecases}

In this \lcnamecref{sec:usecases}, we analyse some satisfaction systems
and establish whether they are (or not) eviction- and
reception-compatible.
The framework we presented in \cref{sec:modelChange} is general enough to cover
several satisfaction systems without imposing much constraints upon the
logics being used to represent an agent's beliefs.
In particular, it covers propositional logic (\cref{propcompat}).
However,  there are interesting fragments of first-order logic used for knowledge representation that
are 
neither eviction nor
reception-compatible,  as it is the case of some DLs
(\cref{ALCfecompat}).
\Cref{tab:compats} summarises the results of compatibility proved in this
\lcnamecref{sec:usecases}.

\begin{table}[tb]
    \centering
    \begin{tabular}{@{}lll@{}}
        \toprule
        \multirow{2}{*}{Satisfaction System} & \multicolumn{2}{c}{Compatible} \\ \cmidrule(l){2-3}
                                             & \Mconnm{}      & \Mexnm{}      \\ \midrule
        \(\logsysprop\)                    & Yes            & Yes           \\
        \(\logsyshorn\)                    & Yes            & Yes           \\
        \(\logsyskleene\)                      & Yes            & Yes           \\
        \(\logsyspriest\)                      & No            & Yes           \\
        \(\logsysfuzzy\)                     & Yes            & Yes           \\
        \(\logsysltlx\)                       & No             & Yes           \\
        \(\logsysDLABox\)                  & Yes             & No            \\
        \(\logsysDLLITE\)\textsuperscript{$\dagger$}                  & Yes             & Yes            \\
                        \(\logsysALC\)                  & No             & No            \\
        \bottomrule
    \end{tabular}
    \caption{\Mconnm- and \mexnm-compatibility of different satisfaction systems\label{tab:compats}. $\dagger$: only with finite signature}
\end{table}


\subsection{The Case of Propositional Logic\label{use:prop}}

We start by analysing the simplest case: that of propositional classical logic.
We denote by
	\(\logsysprop\) the satisfaction system with
	the entailment relation given by the standard
	semantics of propositional logic with finite signature.
As one can express inconsistency with a finite base, 
tautologies, and there is only a finite number of valuations, we obtain the
following result for \(\logsysprop\).
\begin{restatable}{theorem}{propcompat}%
	\label{propcompat}
	\(\logsysprop\)
	is   
	\mexnm-compatible and \mconnm-compatible.
\end{restatable}

\Cref{propOps} demonstrates how to formulate \mconnm{} and \mexnm{} in propositional logic.

\begin{restatable}{proposition}{rPropOps}%
\label{propOps}
The functions \(\evcprop\) and \(\rcpprop\) defined next are, respectively, maxichoice \mconnm{} and \mexnm{} functions on \(\logsysprop\).
\begin{align*}
    \evcprop(\baseb,\mSet) &= \bigvee_{v \in \modelsof{\baseb} \setminus \mSet} \left(\bigwedge_{v(a) = \vtrue} a \wedge \bigwedge_{v(a) = \vfalse} \neg{a} \right)\\
    \rcpprop(\baseb,\mSet) &= \bigvee_{v \in \modelsof{\baseb} \cup \mSet} \left(\bigwedge_{v(a) = \vtrue} a \wedge \bigwedge_{v(a) = \vfalse} \neg{a} \right).
\end{align*}
As usual, \(\vfalse\) stands for `false' and \(\vtrue\) stands for `true'.
\end{restatable}

Horn logic limits the language of propositional logic to only facts and implications.
Let $\propatoms$ be a set of propositional atoms containing $\perp$ (\textit{falsum}), the language of Horn logic, denoted $\llanghorn$, is given 
by the following BNF grammar. 
\begin{align*}
\varphi &:= \varphi \land \varphi \mid H \mid T \to H \\
T &:= T \land T \mid H &
H &:= p
\end{align*}%
where $p \in \propatoms$.

The universe of models and satisfaction system in (propositional) Horn logic coincide with those of classical propositional logic.
The compatibilities of the resulting satisfying system with our setting is given in \cref{horncompat}, which can be proved in a similar way as \cref{propcompat}.

\begin{restatable}{theorem}{rhorncompat}\label{horncompat}
    $\logsyshorn = (\llanghorn, \mUni_{\text{Prop}}, \models_{\text{Prop}})$ $\logsyshorn$, is both eviction- and reception-compatible.
\end{restatable}


 \subsection{The Case of Kleene and Priest 3-valued Logics\label{use:3vl}}

Now, we look at examples of 3-valued logics which are only slightly more complex than propositional logic.
 The 3-valued logics of Kleene~\citep{Kleene1952} and Priest~\citep{Priest} consist of the classical propositional logic 
 in which the formulae might be assigned one of the following three truth values:
true ($\vtrue$), false (\(\vfalse\)) and unknown (\(\vunk\)).
 Consider the following total order on the three values: $\vfalse < \vunk < \vtrue$.
The satisfaction system for Kleene's 3-valued logics is
\(\logsyskleene = (\llangprop,  \mUniKleene,  \models_{K3})\),  and for
Priest's 3-valued logics is \(\logsyspriest = (\llangprop,  \mUniPiest,  \models_{P3})\)
where $\llangprop$ is the language of the classical propositional logic,  and
\begin{itemize}
	\item $\mUni_{3} $ is the set of all functions $v: \llang \to \{\vfalse, \vunk, \vtrue\}$  s.t  
	\begin{itemize}
		\item
			$v(\neg \varphi) = \vtrue$,  if $v(\varphi) = \vfalse$;
			$v(\neg \varphi) = \vunk$,  if $v(\varphi) =\vunk$;
			$v(\neg \varphi) = \vfalse$,  if $v(\varphi) = \vtrue$.
		\item $v(\varphi \land \psi) = \min_<(\{v(\varphi), v(\psi)\})$.
		\item $v(\varphi \lor \psi) = \max_<(\{v(\varphi), v(\psi)\})$.
	\end{itemize}
\end{itemize}

The main difference between Kleene's and Priest's 3-valued logics lies on the satisfaction relation: for Kleene, $v \models_{K3} \varphi$ iff $v(\varphi) = \vtrue$; while for Priest,  $v \models_{P3} \varphi$ iff $v(\varphi) = \vtrue$ or $v(\varphi) = \vunk$.
\Cref{3vlcompat} states the compatibility results for these systems.

\begin{restatable}{theorem}{rthreevlcompat}\label{3vlcompat}
    $\logsyskleene$ and $\logsyspriest$ are reception-compatible but \(\logsyskleene\) is \mexnm-compatible, while \(\logsyspriest\) is not.
\end{restatable}




\subsection{The Case of Propositional Gödel Logic\label{use:fuzzy}}

All satisfaction systems studied earlier in this \lcnamecref{sec:usecases} had only finitely many models.
This is not the case in (propositional) Gödel logic, one of the most important fuzzy logics~\cite{Hajek1998,Bergmann2008}.
We will analyse the compatibilities for Gödel logic's satisfaction system next.

Let \(\theta \in {\left(0, 1\right]}\) and \(\logsysfuzzy = (\llangfuzzy, \mUnifuzzy, \modelsfuzzy)\) be a satisfaction system in which

\begin{itemize}
    \item \(\llangfuzzy\) consists of propositional formulae defined over a non-empty finite set of propositional atoms \(\propatoms\);
    \item \(\mUnifuzzy\) is the set of all functions \(v : \llang \to [0, 1]\) respecting the standard Gödel semantics for the boolean connectives (see~\cite[page 20]{Bergmann2008}); and
    \item \(v \modelsfuzzy \baseb\) iff \(v(\bigwedge_{\varphi \in \baseb \cup \{\neg{(\neg{a} \land a)}\}} \varphi) \geq \theta\), 
     where \(a \in \propatoms\).
\end{itemize}

We say that \(\logsysfuzzy\) is the satisfaction system for propositional Gödel logic with threshold \(\theta\).
\Cref{LFuzzyCompat} states a positive result for \(\logsysfuzzy\).
Despite \(\mUnifuzzy\) being infinite, the models can be grouped into finitely many equivalence classes w.r.t.\ satisfaction of bases.

\begin{restatable}{theorem}{rLFuzzyCompat}\label{LFuzzyCompat}
    The satisfaction system \(\logsysfuzzy\) is \mconnm- and \mexnm-compatible.
\end{restatable}

\subsection{The LTL NeXt Fragment\label{use:ltlnext}}

In the previous \lcnamecrefs{use:ltlnext}, we focused on languages which had only boolean connectives and whose models were valuations on propositional atoms.
Here, we consider the LTL logic \citep{ClarkeBook} with the language confined only to the operator $X$ (NeXt) as an example of satisfaction system which differs considerably in language and in semantics from the other systems presented before.  
For clarity,  the language of this logic $\llangltlx$ is given by the following grammar in BNF
$ \varphi := p \mid X \varphi$,  where $p \in \propatoms$ for some fixed non empty set of  propositional symbols $\propatoms$. 
 We write $X^m p$ as a shorthand for the nesting of $X$ $m$ times.   The formula $X^0p$ stands for $p$.  
A model of this logic is a pair $(M,s)$ where $M$ is a Kripke structure (see definition at \citep{ClarkeBook}), and $s$ is a initial state of $M$,  called the initial state.   
Let $\mUniltlx$ be the set of all such models.  
A model $(M,s)$ satisfies a formula $X^i p$ iff $p$ is labelled at the $i$-th state of all paths from $M$ starting from $s$ (see \citep{ClarkeBook},  for a detailed definition).   Let $\modelsltlx$ be the satisfaction relation between models and formulae as just defined.   
The satisfaction system of this logic is the system $\logsysltlx = (\llang_{X},  \mUni_X,  \modelsltlx)$.  
Within this section,  we will write $A \models_{X} \varphi$ as a shorthand 
for $(M,s) \modelsltlx \varphi$, for all $(M,s) \in A$.  

For \mexnm-compatibility we define the function \(\rcpx\) and prove its relation to the \mexnm{} construction in \cref{prop:ltlxrecep}.


\begin{restatable}{proposition}{rltlxrecep}\label{prop:ltlxrecep}
    Let \(\baseb \in \finitepwset(\llangltlx)\), \(\mSet \subseteq \mUniltlx\) and
    $\rcpx: \finitepwset(\llangltlx) \times \powerset(\mUni_X) \to \finitepwset(\llangltlx)$ defined as 
    \[ \rcpx(\baseb, \mSet) = \{ \varphi \in \baseb \mid \mSet \models \varphi \} .\]
    It holds that $\rcpx(\baseb, \mSet) \in \FRsups(\modelsof{\baseb} \cup \mSet)$.  
\end{restatable}

Even though this logic is \mexnm-compatible, it is not \mconnm-compatible.

\begin{restatable}{theorem}{rltlxcompats}\label{ltlxcompats}
    $\logsysltlx$ is \mexnm-compatible but it is not \mconnm-compatible.
\end{restatable}

\subsection{The Case of Description Logic 
	\label{use:alc}}

To analyse the case of Description Logic (DL),
we study  \(\ALC\), which is a prototypical DL that shares many  similarities with other expressive logics in the DL family. Here we use the term \emph{ontology} to refer to a finite set of formulae---a finite base.
Let \NC{}, \NR{} and \NI{} be countably infinite and pairwise disjoint sets of concept, role, and individual names, respectively.
\ALC{}
\emph{concepts} are built according to the rule:
\[
C::=  A \mid \neg C \mid (C \sqcap C) \mid \exists r.C,
\]
where $A \in \NC$.
An \(\ALC{}\) \emph{ontology}
is a set of expressions 
of the form
\[C(a) \mid r(a,b) \mid C \sqsubseteq D,\]
%
where $C,D$ are \ALC{} concepts, $a, b \in \NI$, and $r \in \NR$. 
	The semantics of the DLs considered here 
	is 
	standard~\cite{baader_horrocks_lutz_sattler_2017}.

\begin{restatable}{theorem}{allALCcompat}%
	\label{ALCfecompat}
	\(\logsysALC\) is neither 
	\mexnm-compatible nor \mconnm-compatible.
\end{restatable}

Not being \mexnm-compatible is essentially due
to having an infinite signature. Indeed,
this is already the case for
 the satisfaction system where the language
allows only (positive and negative) \emph{assertions}, which are expressions of the form $A(a),r(a,b),\neg A(a),\neg r(a,b)$, where $A\in\NC$, $r\in\NR$, and $a,b\in\NI$. We denote it by \(\logsysDLABox\).

\begin{restatable}{theorem}{allABoxcompat}%
	\label{dlaboxcompat}
	\(\logsysDLABox\) is  not 
	\mexnm-compatible but it is \mconnm-compatible.
\end{restatable}


Finally, we  consider the case in which
the signature is finite, that is, the sets
$\NC, \NR, \NI$ are
disjoint, non-empty, and \emph{finite} (but models can still be infinite).
Our result that \(\logsysALC\)  is not \mconnm-compatible already holds is this case.
So we consider a simpler but popular DL called DL-Lite$_\Rmc$.
DL-Lite$_\Rmc$ \emph{role and concept
	inclusions} are expressions of the form $S \sqsubseteq T$ and $B \sqsubseteq C$,
respectively, where $S, T$ are role expressions and $B, C$ are
concept expressions built through the  rules
\[S ::= r \mid r^-, T ::= S \mid \neg S, B ::= A \mid \exists S, C ::= B \mid \neg B,\]
with $r \in \NR$ and $A \in \NC$. A \emph{DL-Lite$_\Rmc$ ontology} is a set of role and concept inclusions and (positive) assertions, as defined above.
We denote by \(\logsysDLLITE\) the satisfaction system
with the entailment relation given by the standard semantics of DL-Lite$_\Rmc$~\citep{baader_horrocks_lutz_sattler_2017}.
\begin{restatable}{theorem}{allDLLitecompat}%
	\label{dllitecompat}
	\(\logsysDLLITE\) (with finite signature) is  
	\mexnm-compatible and
	\mconnm-compatible.
\end{restatable}

\section{Related Work}\label{sec:relatedWorks}

%

Finite representation of epistemic states have been addressed in Belief Change literature by representing an agent's knowledge via a finite set of formulae known as a finite belief base \citep{Nebel1991,Dixon1993}. %
Belief change operations on belief bases,  however,  are syntax sensitive: they preserve the syntactic form of the original belief base as much as possible.  
This syntax sensitivity also appears in traditional approaches for Ontology Repair and Ontology Evolution \citep{Kalyanpur2006,Suntisrivaraporn2009}.  
Although finite bases trivially guarantee finite representability, syntax sensitivity might compel drastic loss of information as noticed by \citet{Hansson1993}.  
The main reason is that applying an operation in the finite base is not equivalent to applying an operation on the epistemic state generated by the same base, in general.
The new paradigm we defined performs \mconnm{} and \mexnm{} on the epistemic state generate from the finite base, that is, it is not sensitive to syntax.
The problem of loss of information due to syntax sensitivity has been studied in Belief Change pseudo-contraction~\citep{Santos2018}. 
Thus, our paradigm approaches the concept of pseudo-contraction with the extra condition of finite representability.  

To minimize the drastic loss of syntax sensitive operations, \citet{Troquard2018} proposed to repair DL ontologies by
weakening axioms using refinement operators.  
Building on this study, \citet{Baader2018} devised the theory of \emph{gentle repairs}, which also aims at
keeping most of the information within the ontology upon repair.  In fact,  
gentle repairs are type of pseudo-contractions~\citep{Matos2019}. In
this same category, we include the Belief Change operations
based on concept relaxation~\citep{Aiguier2018}.
These studies, however, do not answer the question of finding an optimal solution.
Meanwhile, we give conditions that guarantee that our operations perform minimal changes on epistemic states.
\Citet{Baader2022} propose to repair \(\EL\) ontologies by modifying only their ABox,   
 preserving as many entailments as possible.  
Still, in this approach,  one cannot contract all necessary kinds of information, as the TBox cannot be modified.  



Other works in Belief Change that consider finite representability are:  
(i) revision by \citet{Katsuno1991} and (ii) base-generated operations by \citet{Hansson1996}.  In
the former, \Citet{Katsuno1991} assumes an agent's epistemic state is represented as 
a single formula.  
This is possible because they only
consider finitary propositional languages. \Citet{Hansson1996} provides
a characterisation of Belief Change operations over finite bases but restricted
for logics which satisfy all the AGM assumptions (such as classical propositional logic), while we have shown that our approach works in other logics as well.
%
%

As for Belief Change operation on models,  
\Citet{Guerra2019} consider modifying a single Kripke model into a new one 
that satisfies a given formula in 
Linear Temporal Logics (LTL) \citep{ClarkeBook}.
While they provide an AGM-style characterisation, there is no guarantee of finite representability.
\Citet{Hieke2021} devise an approach for contraction by formula in DL \(\EL\) ontologies that employs the notion of counter-models.
Even so, while a model is employed to derive the final outcome of the contraction, the input is still a single formula.
Hence, despite using finite bases, our framework is more general because we accept arbitrary sets of models as input.

\section{Conclusion and Future Work}%
\label{sec:conc}

We introduced a new paradigm of Belief Change: an agent's
epistemic state is represented as a finite base,  while incoming information are
represented as a set of models.
The agent can either incorporate the  incoming
models (via \mexnm{}) or remove them (via \mconnm{}).
In either case,  the resulting belief base must be finitely representable.
%
The standard rationality postulates of Belief Change do not guarantee finite representability.
Hence, we proposed new postulates that capture a notion of minimal change in this setting for both \mconnm{} and \mexnm{}.
We also presented two constructive classes of model change operations that are precisely
characterised by such sets of rationality postulates.
As a case study, we  investigated how this new paradigm works in various logics.
%


Eviction can lead to an inconsistent belief
base,  in the case that all models are removed.  If consistency is required,
then a more sophisticated model operation could be defined with the caveat that,
in behalf of consistency,  other models can be assimilated during the removal of an input model.  This third
model operation is similar in spirit to formula revision.
We leave model revision as a future work.  We envisage that the
results we obtain for \mconnm{} and \mexnm{} shall shed light towards this other
operation.
%
Another line of research concerns
the effects of partially constraining the structure of the resulting base,
in the spirit of
pseudo-contractions.

\section*{Acknowledgements}

Part of this work has been done in the context of CEDAS (Center for Data
Science, University of Bergen, Norway).
The first author is supported by the ERC project ``Lossy Preprocessing''
(LOPRE), grant number 819416, led by Prof Saket Saurabh.
The second author is supported by the NFR project ``Learning Description Logic
Ontologies'', grant number 316022.
The third author is supported by the
German Research Association (DFG), project number 424710479.

\bibliography{references.bib}

\begin{thebibliography}{37}
\providecommand{\natexlab}[1]{#1}

\bibitem[{Aiguier et~al.(2018)Aiguier, Atif, Bloch, and Hudelot}]{Aiguier2018}
Aiguier, M.; Atif, J.; Bloch, I.; and Hudelot, C. 2018.
\newblock Belief revision, minimal change and relaxation: A general framework
  based on satisfaction systems, and applications to description logics.
\newblock \emph{Artificial Intelligence}, 256: 160--180.

\bibitem[{Alchourr{\'{o}}n, G{\"{a}}rdenfors, and
  Makinson(1985)}]{Alchourron1985}
Alchourr{\'{o}}n, C.~E.; G{\"{a}}rdenfors, P.; and Makinson, D. 1985.
\newblock {On the Logic of Theory Change: Partial Meet Contraction and Revision
  Functions}.
\newblock \emph{Journal of Symbolic Logic}, 50(2): 510--530.

\bibitem[{Arias, Khardon, and Maloberti(2007)}]{Arias2007}
Arias, M.; Khardon, R.; and Maloberti, J. 2007.
\newblock Learning Horn Expressions with {LOGAN-H}.
\newblock \emph{J. Mach. Learn. Res.}, 8: 549--587.

\bibitem[{Baader et~al.(2017)Baader, Horrocks, Lutz, and
  Sattler}]{baader_horrocks_lutz_sattler_2017}
Baader, F.; Horrocks, I.; Lutz, C.; and Sattler, U. 2017.
\newblock \emph{An Introduction to Description Logic}.
\newblock Cambridge University Press.

\bibitem[{Baader et~al.(2022)Baader, Koopmann, Kriegel, and
  Nuradiansyah}]{Baader2022}
Baader, F.; Koopmann, P.; Kriegel, F.; and Nuradiansyah, A. 2022.
\newblock {Optimal {ABox} Repair w.r.t. Static EL TBoxes: From Quantified
  ABoxes Back to ABoxes}.
\newblock In \emph{The Semantic Web}. Springer International Publishing.

\bibitem[{Baader et~al.(2018)Baader, Kriegel, Nuradiansyah, and
  Pe{\~{n}}aloza}]{Baader2018}
Baader, F.; Kriegel, F.; Nuradiansyah, A.; and Pe{\~{n}}aloza, R. 2018.
\newblock {Making Repairs in Description Logics More Gentle}.
\newblock In \emph{{KR} 2018}. {AAAI} Press.

\bibitem[{Bergmann(2008)}]{Bergmann2008}
Bergmann, M. 2008.
\newblock \emph{{An Introduction to Many-Valued and Fuzzy Logic}}.
\newblock Cambridge University Press.

\bibitem[{Clarke et~al.(2018)Clarke, Grumberg, Kroening, Peled, and
  Veith}]{ClarkeBook}
Clarke, E.~M.; Grumberg, O.; Kroening, D.; Peled, D.~A.; and Veith, H. 2018.
\newblock \emph{Model checking, 2nd Edition}.
\newblock {MIT} Press.
\newblock ISBN 978-0-262-03883-6.

\bibitem[{Dalal(1988)}]{Dalal88}
Dalal, M. 1988.
\newblock Investigations into a Theory of Knowledge Base Revision.
\newblock In \emph{Proceedings of the 7th National Conference on Artificial
  Intelligence}, 475--479. {AAAI} Press / The {MIT} Press.

\bibitem[{{De Raedt}(1997)}]{DeRaedt1997}
{De Raedt}, L. 1997.
\newblock Logical settings for concept-learning.
\newblock \emph{Artificial Intelligence}, 95(1): 187--201.

\bibitem[{Delgrande, Peppas, and Woltran(2018)}]{DelgrandePW18}
Delgrande, J.~P.; Peppas, P.; and Woltran, S. 2018.
\newblock General Belief Revision.
\newblock \emph{J. {ACM}}, 65(5): 29:1--29:34.

\bibitem[{Delgrande and Wassermann(2010)}]{DelgrandeW10}
Delgrande, J.~P.; and Wassermann, R. 2010.
\newblock Horn Clause Contraction Functions: Belief Set and Belief Base
  Approaches.
\newblock In \emph{Principles of Knowledge Representation and Reasoning:
  Proceedings of the Twelfth International Conference, {KR} 2010, Toronto,
  Ontario, Canada, May 9-13, 2010}. {AAAI} Press.

\bibitem[{Delgrande and Wassermann(2013)}]{DelgrandeW13}
Delgrande, J.~P.; and Wassermann, R. 2013.
\newblock Horn Clause Contraction Functions.
\newblock \emph{J. Artif. Intell. Res.}, 48: 475--511.

\bibitem[{Dixon and Wobcke(1993)}]{Dixon1993}
Dixon, S.; and Wobcke, W. 1993.
\newblock {The Implementation of a First-Order Logic AGM Belief Revision
  System}.
\newblock In \emph{{ICTAI} 1993}, 40--47. {IEEE} Computer Society.

\bibitem[{Dixon(1994)}]{dix94}
Dixon, S.~E. 1994.
\newblock \emph{Belief revision: A computational approach}.
\newblock Ph.D. thesis, University of Sydney.

\bibitem[{Guerra and Wassermann(2019)}]{Guerra2019}
Guerra, P.~T.; and Wassermann, R. 2019.
\newblock {Two AGM-style characterizations of model repair}.
\newblock \emph{Ann. Math. Artif. Intell.}, 87(3): 233--257.

\bibitem[{Hansson(1997)}]{Hansson1997}
Hansson, S. 1997.
\newblock Semi-revision.
\newblock \emph{Journal of Applied Non-Classical Logics}, 7(1-2): 151--175.

\bibitem[{Hansson(1993{\natexlab{a}})}]{Hansson1993}
Hansson, S.~O. 1993{\natexlab{a}}.
\newblock Changes of disjunctively closed bases.
\newblock \emph{Journal of Logic, Language and Information}, 2(4): 255--284.

\bibitem[{Hansson(1993{\natexlab{b}})}]{Hansson93}
Hansson, S.~O. 1993{\natexlab{b}}.
\newblock Reversing the Levi identity.
\newblock \emph{J. Philos. Log.}, 22(6): 637--669.

\bibitem[{Hansson(1996)}]{Hansson1996}
Hansson, S.~O. 1996.
\newblock {Knowledge-Level Analysis of Belief Base Operations}.
\newblock \emph{Artificial Intelligence}, 82(1-2): 215--235.

\bibitem[{Hansson(1999)}]{Hansson1999}
Hansson, S.~O. 1999.
\newblock \emph{{A Textbook of Belief Dynamics: Theory Change and Database
  Updating}}.
\newblock Applied Logic Series. Kluwer Academic Publishers.

\bibitem[{Hieke, Kriegel, and Nuradiansyah(2021)}]{Hieke2021}
Hieke, W.; Kriegel, F.; and Nuradiansyah, A. 2021.
\newblock Repairing $\mathcal{EL}$ {TBoxes} by Means of Countermodels Obtained
  by Model Transformation.
\newblock In Homola, M.; Ryzhikov, V.; and Schmidt, R.~A., eds.,
  \emph{Proceedings of the 34th International Workshop on Description Logics
  {(DL} 2021), Bratislava, Slovakia, September 19-22, 2021}, volume 2954 of
  \emph{{CEUR} Workshop Proceedings}. CEUR-WS.org.

\bibitem[{Hájek(1998)}]{Hajek1998}
Hájek, P. 1998.
\newblock \emph{{Metamathematics of Fuzzy Logic}}.
\newblock Springer Netherlands.
\newblock ISBN 9789401153003.

\bibitem[{Kalyanpur(2006)}]{Kalyanpur2006}
Kalyanpur, A. 2006.
\newblock \emph{Debugging and repair of {OWL} ontologies}.
\newblock Ph.D. thesis, University of Maryland.

\bibitem[{Katsuno and Mendelzon(1991)}]{Katsuno1991}
Katsuno, H.; and Mendelzon, A.~O. 1991.
\newblock Propositional knowledge base revision and minimal change.
\newblock \emph{Artificial Intelligence}, 52(3): 263--294.

\bibitem[{Kleene(1952)}]{Kleene1952}
Kleene, S. 1952.
\newblock \emph{Introduction to Metamathematics}.
\newblock Princeton, NJ, USA: North Holland.

\bibitem[{Matos et~al.(2019)Matos, Guimar{\~{a}}es, Santos, and
  Wassermann}]{Matos2019}
Matos, V.~B.; Guimar{\~{a}}es, R.; Santos, Y.~D.; and Wassermann, R. 2019.
\newblock {Pseudo-contractions as Gentle Repairs}.
\newblock In \emph{Lecture Notes in Computer Science}, 385--403. Springer
  International Publishing.

\bibitem[{Nebel(1990)}]{Nebel90}
Nebel, B. 1990.
\newblock \emph{Reasoning and Revision in Hybrid Representation Systems},
  volume 422 of \emph{Lecture Notes in Computer Science}.
\newblock Springer.

\bibitem[{Nebel(1991)}]{Nebel1991}
Nebel, B. 1991.
\newblock {Belief Revision and Default Reasoning: Syntax-Based Approaches}.
\newblock In \emph{KR 1991}, 417--428. Morgan Kaufmann.

\bibitem[{Priest(1979)}]{Priest}
Priest, G. 1979.
\newblock The Logic of Paradox.
\newblock \emph{Journal of Philosophical Logic}, 8(1): 219--241.

\bibitem[{Ribeiro, Nayak, and Wassermann(2018)}]{RibeiroNW18}
Ribeiro, J.~S.; Nayak, A.; and Wassermann, R. 2018.
\newblock {Towards Belief Contraction without Compactness}.
\newblock In \emph{{KR} 2018}, 287--296. {AAAI} Press.

\bibitem[{Ribeiro, Nayak, and Wassermann(2019{\natexlab{a}})}]{RibeiroNW19AAAI}
Ribeiro, J.~S.; Nayak, A.; and Wassermann, R. 2019{\natexlab{a}}.
\newblock {Belief Change and Non-Monotonic Reasoning Sans Compactness}.
\newblock In \emph{AAAI 2019}, 3019--3026. {AAAI} Press.

\bibitem[{Ribeiro, Nayak, and Wassermann(2019{\natexlab{b}})}]{RibeiroNW19}
Ribeiro, J.~S.; Nayak, A.; and Wassermann, R. 2019{\natexlab{b}}.
\newblock {Belief Update without Compactness in Non-finitary Languages}.
\newblock In \emph{{IJCAI} 2019}, 1858--1864. ijcai.org.

\bibitem[{Ribeiro(2013)}]{Ribeiro2013}
Ribeiro, M.~M. 2013.
\newblock \emph{{Belief Revision in Non-Classical Logics}}.
\newblock Springer London.

\bibitem[{Santos et~al.(2018)Santos, Matos, Ribeiro, and
  Wassermann}]{Santos2018}
Santos, Y.~D.; Matos, V.~B.; Ribeiro, M.~M.; and Wassermann, R. 2018.
\newblock Partial meet pseudo-contractions.
\newblock \emph{International Journal of Approximate Reasoning}, 103: 11--27.

\bibitem[{Suntisrivaraporn(2009)}]{Suntisrivaraporn2009}
Suntisrivaraporn, B. 2009.
\newblock \emph{Polynomial time reasoning support for design and maintenance of
  large-scale biomedical ontologies}.
\newblock Ph.D. thesis, Dresden University of Technology, Germany.

\bibitem[{Troquard et~al.(2018)Troquard, Confalonieri, Galliani,
  Pe{\~{n}}aloza, Porello, and Kutz}]{Troquard2018}
Troquard, N.; Confalonieri, R.; Galliani, P.; Pe{\~{n}}aloza, R.; Porello, D.;
  and Kutz, O. 2018.
\newblock {Repairing Ontologies via Axiom Weakening}.
\newblock In \emph{AAAI 2018}, 1981--1988. {AAAI} Press.

\end{thebibliography}

\appendix
\clearpage
\pagebreak

\renewcommand{\thelemma}{\Alph{section}.\arabic{lemma}}
\setcounter{lemma}{0}

\renewcommand{\thedefinition}{\Alph{section}.\arabic{definition}}
\setcounter{definition}{0}

\renewcommand{\theproposition}{\Alph{section}.\arabic{proposition}}
\setcounter{proposition}{0}

\renewcommand{\theobservation}{\Alph{section}.\arabic{observation}}
\setcounter{observation}{0}

\renewcommand{\thecorollary}{\Alph{section}.\arabic{corollary}}
\setcounter{corollary}{0}

\renewcommand{\thetheorem}{\Alph{section}.\arabic{theorem}}
\setcounter{theorem}{0}

\section{Proofs for \Cref{sec:modelChange}}






\noPM*

\begin{proof}
    As an example, we consider the satisfaction \(\Lambda_t
    = (\llang_t, \mUni_t, \models_t)\) where: \(\llang_t = \{a,
    b\}\) with \(a, b\) being propositional atoms; \(\mUni_t\) the boolean valuations (\(\vtrue\) for `true' and \(\vfalse\) for `false') to the pair \((a,
    b)\), and the satisfaction relation \(\models_t\) defined as usual.

    We have that 
    \begin{align*}
        \FRsubs(\{(\vtrue, \vtrue), (\vtrue, \vfalse), (\vfalse, \vtrue)\}, \Lambda_\land) =  \\ 
        \qquad \{\{(\vfalse, \vtrue), (\vtrue, \vtrue)\}, \{(\vtrue, \vfalse), (\vtrue, \vtrue)\}\}.
    \end{align*}
     
    The intersection of the resulting subsets is \(\{(\vtrue, \vtrue)\}\), which
    cannot be represented in \(\Lambda_t\).
\end{proof}

\noPMR*

\begin{proof}
 As an example, we consider the satisfaction \(\Lambda_p
    = (\llang_p, \mUni_p, \models_p)\) where: \(\llang_p = \{\perp,  a,
    b\}\) with \(a, b\) being propositional atoms and \(\perp\) \textit{falsum}; 
\(\mUni_p\) the boolean valuations (\(\vtrue\) for `true' and \(\vfalse\) for `false') to the pair \((a,
    b)\), and the satisfaction relation \(\models_p\) defined as usual.

    We have that 
    \begin{align*}
        \FRsups(\{(\vtrue, \vtrue)\}, \Lambda_p) =  \\ 
        \qquad \{\{  (\vtrue, \vtrue), (\vtrue, \vfalse)\}, \{(\vtrue, \vtrue),  (\vfalse, \vtrue)\}\}.
    \end{align*}
     
    The union of the resulting supersets is the set  $ \{\{  (\vtrue, \vtrue) (\vtrue, \vfalse),   (\vfalse, \vtrue)\}\}$ , which
    cannot be represented in \(\Lambda_p\),  as we cannot express disjunction.  
\end{proof}

To prove the representation theorem for \mconnm{},  we will need some auxiliary tools.  Recall that we write $\powerset^*(A)$ as a shorthand for $\powerset( 
A) \setminus \{\emptyset\}$, that is,  the power set of $A$ without the empty
set $\emptyset$.  Given an \mconnm{} function  $\mCon$ on an \mconnm-compatible satisfaction system $\logsys$,  we define the function $\xi^-: \powerset^*(\FRsets(\logsys)) \to \finitepwset(\llang) \times \powerset(\mUni) $,   such that 
\[ \xi^-(X) = \{  (\baseb,  \mSet)  \mid  \FRsubs(\modelsof{\baseb} \setminus \mSet, \logsys) = X \}.\]

Intuitively,  $\xi^-(X)$ holds all the pairs $(\baseb, \mSet)$ such that $X$
contains exactly all finite representable sets of models closest to
$\modelsof{\baseb} \setminus \mSet$.  
We  also define the function $\mathcal{C}^-:  \powerset^*(\FRsets(\logsys)) \to \powerset(\mUni)$ such that
\[\mathcal{C}^-(X) = \{ \modelsof{\mCon(\baseb, \mSet)} \mid  (\baseb, \mSet) \in \xi^-(X) \}.\]

%
%

\begin{lemma}\label{lem:app:modevecxi}
Let $\logsys$ be an \mconnm-compatible satisfaction system. If a model change operation $\mCon$ satisfies \textit{uniformity} then for all $X \in \powerset(\FRsets(\logsys))$: 

\begin{enumerate}[label= (\roman*), leftmargin=*]
    \item $\modelsof{\mCon(\baseb, \mSet)} = \modelsof{\mCon(\baseb', \mSet')}$
for all $(\baseb, \mSet),  (\baseb', \mSet') \in \xi^-(X)$;
    \item $\mathcal{C}^-(X)$ is a singleton,  if $\xi^-(X) \neq \emptyset$.
\end{enumerate}
\begin{proof}
Let $\mCon$ be a model change operation satisfying \textit{uniformity},  and $X
\in \FRsets(\logsys)$, where $\logsys$ is an \mconnm-compatible satisfaction system.  
\begin{enumerate}[label= (\roman*), leftmargin=*]
	\item Let $(\baseb, \mSet),  (\baseb', \mSet') \in \xi^-(X)$.  Thus, by definition of $\xi^-$, we have that: 
        \begin{align*} 
            X &= \FRsubs(\modelsof{\baseb} \setminus \mSet, \logsys)\\ 
            {} &= \FRsubs(\modelsof{\baseb'} \setminus \mSet', \logsys).
        \end{align*}

Hence, from \textit{uniformity}, we get 
\[ \modelsof{\mCon(\baseb, \mSet)}  =  \modelsof{\mCon(\baseb', \mSet')}.\]

    \item Let us suppose that $\xi^-(X) \neq \emptyset$. Let us fix such
        a $(\baseb, \mSet) \in \xi^-(X)$.  
	By definition of $\xi^-$, we have that 
	\[ X = \FRsubs(\modelsof{\baseb} \setminus \mSet, \logsys).\]
	By definition of $\mathcal{C}^-$:  
\[ \modelsof{\mCon(\baseb,\mSet)} \in \mathcal{C}^-(X).\]	

Hence, to show that $\mathcal{C}^-(X)$ is a singleton,  we need to prove
that:  $Y = \modelsof{\mCon(\baseb,\mSet)}$,  for all $Y \in \mathcal{C}^-(X)$.   Let $Y \in \mathcal{C}^-(X)$.  By definition of $\mathcal{C}^-$, we have that for some \((\baseb', \mSet') \in \xi^-(X)\) it holds that 
\(Y = \modelsof{\mCon(\baseb', \mSet')}\).
Thus,  as both pairs $(\baseb, \mSet),  (\baseb', \mSet') \in \xi^-(X)$,  we get from item (i) above that $\modelsof{\mCon(\baseb, \mSet)} = \modelsof{\mCon(\baseb', \mSet')}$.  Therefore,  
$Y = \modelsof{\mCon(\baseb, \mSet)}$.  This concludes the proof.\qedhere   
\end{enumerate}
\end{proof}
\end{lemma}

\begin{proposition}\label{obs:_finiteMAXFR}
Let $\logsys = (\llang, \mUni, \models)$ be an \mconnm-compatible satisfaction system.  
If a model change operation $\mCon$ satisfies \textit{success, inclusion} and
\textit{finite retainment},  then $\modelsof{\mCon(\baseb, \mSet)} \in
\FRsubs(\modelsof{\baseb} \setminus \mSet, \logsys)$ for all \(\baseb \in
\finitepwset(\llang)\) and \(\mSet \subseteq \mUni\). 
\begin{proof}
Let us suppose for contradiction that there is a model change operation that
satisfies  \textit{success, inclusion} and \textit{finite retainment},  
but $\modelsof{\mCon(\baseb, \mSet)} \not \in \FRsubs(\modelsof{\baseb}
\setminus \mSet, \logsys)$,  for some finite base $\baseb$ and set of models $\mSet$.  
Let us fix such a base $\baseb$ and set $\mSet$. 

From \textit{success} and \textit{inclusion}, we have that
 \[\modelsof{\mCon(\baseb, \mSet)} \subseteq \modelsof{\baseb} \setminus \mSet. \]
 
 By construction,  $\mCon(\baseb, \mSet)$ is a finite base,  which means 
 \begin{align}
 \modelsof{\mCon(\baseb, \mSet)} \in \FRsets(\logsys). \label{eq:mConsys}
 \end{align}
  We know that $\FRsubs(\modelsof{\baseb} \setminus
  \mSet, \logsys) \neq \emptyset$ as $\logsys$ is \mconnm-compatible.  
%
  Let 
  \[Y = \{ X \in \FRsets(\logsys) \mid X \subseteq (\modelsof{\baseb} \setminus \mSet)\}. \] 
  $\modelsof{\mCon(\baseb, \mSet)} \not \in
  \FRsubs(\modelsof{\baseb} \setminus \mSet, \logsys)$ from hypothesis, which means that either $\modelsof{\mCon(\baseb, \mSet)} \not \in \FRsets(\logsys)$ or 
  $\modelsof{\mCon(\baseb, \mSet)}$ is not $\subseteq$-maximal within 
  $Y$.  
  This fact combined with \cref{eq:mConsys} implies that 
  $\modelsof{\mCon(\baseb, \mSet)}$ is not $\subseteq$-maximal within $Y$.  
  Therefore,   there is some $\mSet' \in \FRsets(\logsys)$ such that $\mSet' \subseteq (\modelsof{\baseb} \setminus \mSet)$ and  $\modelsof{\mCon(\baseb, \mSet)} \subset \mSet'$. 
  But \textit{finite retainment} states that $\mSet' \not \in \FRsets(\logsys)$, which is a contradiction.\qedhere 
\end{proof}
\end{proposition}

\mFRConRepr*

\begin{proof}
    ``$\Rightarrow$''    Let \(\mCon_\FRsel\) be a maxichoice \mconnm{} 
    function over \(\logsys\) based on a \frselnm{} \(\FRsel\),  and $\mSet$ be a set of models.

    The function \(\mCon_\FRsel\) satisfies success and inclusion  since \(\modelsof{\mCon_\FRsel(\baseb, \mSet)} \subseteq \FRsubs(\modelsof{\baseb}
    \setminus \mSet, \logsys)\).  
    
    \begin{description}[leftmargin=*]
		\item[(vacuity)]    Assume that $\mSet \cap \modelsof{\baseb} = \emptyset$.  
	Then,    $\FRsubs(\modelsof{\baseb}\setminus \mSet, \logsys)= \{ \{ \modelsof{\baseb}\} \}$, which implies that 
	\[\mCon(\baseb, \mSet) = \FRsel(\FRsubs(\modelsof{\baseb}\setminus \mSet, \logsys)) = \baseb', \] 
	such that $\baseb' = \modelsof{\baseb}$.  
	Therefore,  $\modelsof{\mCon(\baseb, \mSet)} = \modelsof{\baseb}$.   
    
    	\item[(finite retainment)]  Suppose that   
    \(\modelm' \in \modelsof{\baseb} \setminus \modelsof{\mCon_\FRsel(\baseb,
    \mSet)}\), then, by construction, there is no \(\mSet' \in \FRsets(\logsys)\) which
    contains \(\modelm'\) and \(\modelsof{\mCon_\FRsel(\baseb,
    \mSet)} \subset \mSet'\).  
    
    	\item[(uniformity)]  Let  $\FRsubs(\modelsof\baseb \setminus \mSet,  \logsys) =\FRsubs(\modelsof{\baseb'} \setminus \mSet',  \logsys)$.  By definition,  
    	\begin{align*}
    	\modelsof{\mCon(\baseb,\mSet)} &= \FRsel(\FRsubs(\modelsof\baseb \setminus \mSet, \logsys))\\
        \text{and}\\
    	\modelsof{\mCon(\baseb',\mSet')} &= \FRsel(\FRsubs(\modelsof{\baseb'} \setminus \mSet', \logsys)).
    	\end{align*}
    	Therefore, as $\FRsubs(\modelsof{\baseb} \setminus \mSet,  \logsys) =\FRsubs(\modelsof{\baseb'} \setminus \mSet',  \logsys)$,  we can conclude that $\modelsof{\mCon(\baseb,\mSet)} = \modelsof{\mCon(\baseb',\mSet')}$.  

    \end{description}
    Hence, every maxichoice eviction
    function based on a \frselnm{} satisfies all
    postulates stated.

``$\Leftarrow$''    Let \(\mCon : \finitepwset(\llang) \times \powerset(\mUni) \to
    \finitepwset(\llang)\) be a function satisfying the postulates stated. 
	As $\mCon$ satisfies uniformity,  we known from \cref{lem:app:modevecxi} that $\mathcal{C}^-(X)$ is a singleton for every $X \in \powerset(\FRsets(\logsys))$.  Thus,  we can construct the function 
	$\FRsel: \powerset^*(\FRsets(\logsys)) \to \FRsets(\logsys)$ such that
		\begin{align*}
            \FRsel(X) = 
            \begin{dcases}
                Z \text{ s.t. } \mathcal{C}^-(X) = \{Z\} & \text{if } \xi^-(X) \neq \emptyset,\\
                Y \mbox{ s.t.  } Y \in X & \text{otherwise.}
            \end{dcases}
		\end{align*}
	We will prove that: (i) $\FRsel$ is indeed a \frselnm{},  and (ii) 
	that $\modelsof{\mCon(\baseb, \mSet)} = \FRsel(\FRsubs(\modelsof{\baseb} \setminus \mSet, \logsys))$.
	
\begin{enumerate}[label= (\roman*), leftmargin=*]
	\item \textit{$\mathbf{\FRsel}$ is indeed a selection function.  }
        Let $X \in \nepwset(\FRsets(\logsys))$.
		We only need to show that $\FRsel(X) \in X$.  
		The case that $\xi^-(X) = \emptyset$  is trivial,  as $\FRsel$ chooses an arbitrary $Y \in X$ (by the axiom of choice).  Let us focus on the case $\xi^-(X) \neq \emptyset$.  
		From above,  we have that $\FRsel(X) = Z$,  where $\mathcal{C}^-(X) = \{Z \}$.  
		By definition of $\mathcal{C}^-$,  we have that
		 there is a pair  $(\baseb, \mSet) \in \xi^-(X)$ such that
		\[ Z = \modelsof{\mCon(\baseb,\mSet)}. \]
		Let us fix such a $(\baseb, \mSet) \in \xi^-(X)$.  
		Thus,  by definition of $\xi^-$, we get
		\[ X = \FRsubs(\modelsof{\baseb} \setminus \mSet, \logsys).\]
		
Additionally, we know that $\modelsof{\mCon(\baseb,\mSet)} \in  \FRsubs(\modelsof{\baseb} \setminus \mSet, \logsys)$ as a consequence of \cref{obs:_finiteMAXFR}. Thus,  from the identities above we get that $Z \in X$,  which means $\FRsel(X) \in X$. 
%
%
		
	\item $\modelsof{\mCon(\baseb, \mSet)} = \FRsel(\FRsubs(\modelsof{\baseb} \setminus \mSet, \logsys))$.  
		Let $X = \FRsubs(\modelsof{\baseb}\setminus \mSet, \logsys)$.  
        We know that \(X \neq \emptyset\) due to \mconnm-compatibility.
			By definition of $\xi$, we get that 
			\[(\baseb, \mSet) \in \xi^-(X). \]
	%
			By construction,  we have that $\FRsel(X) = Z$ such that  	$\mathcal{C}^-(X) = \{ Z\}$,  which implies  from definition of $\mathcal{C}^-$ that $Z = \modelsof{\mCon(\baseb', \mSet')}$,  for some $(\baseb', \mSet') \in \xi^-(X)$.  
			Therefore,  as 	$(\baseb, \mSet) \in \xi^-(X)$, we get 
			from \cref{lem:app:modevecxi},  that for all  $(\baseb', \mSet') \in \xi^-(X),  $
		\[  \modelsof{\mCon(\baseb, \mSet)} = \modelsof{\mCon(\baseb', \mSet')}. \]
	Thus, 
$\modelsof{\mCon(\baseb, \mSet)} = Z$.  
As $\FRsel(X) = Z$ and $X = \FRsubs(\modelsof{\baseb}\setminus \mSet, \logsys)$, we have that 
\(\modelsof{\mCon(\baseb, \mSet)}\) is equal to \(\FRsel(\FRsubs(\modelsof{\baseb} \setminus \mSet, \logsys))\).\qedhere
\end{enumerate}	

\end{proof}

\vacuity*
\begin{proof}
Assume that $\mCon$ satisfies finite retainment and inclusion,  and that $\modelsof{\baseb} \cap \mSet = \emptyset$.  This means that $\modelsof{\baseb}$ is the closest finite representable set of models disjoint with $\mSet$.  
From inclusion,  $\modelsof{\mCon(\baseb, \mSet)} \subseteq \modelsof{\baseb}$.  
Thus,  from finite retainment,  we get that $\modelsof{\mCon(\baseb, \mSet)} = \modelsof{\baseb}$.  
\end{proof}

To prove the representation theorem for maxichoice \mexnm{},  we will need some auxiliary tools.
The auxiliary tools are analagous to the ones defined for the representation theorem of eviction.  
Given a \mexnm{} function  $\mExp$ on a \mexnm-compatible satisfaction system $\logsys$,  we define the function $\xi^+: \powerset^*(\FRsets(\logsys)) \to \finitepwset(\llang) \times \powerset(\mUni) $ as 
\[ \xi^+(X) = \{  (\baseb,  \mSet)  \mid  \FRsubs(\modelsof{\baseb} \cup \mSet, \logsys) = X \}.\]

Intuitively,  $\xi^+(X)$ holds all the pairs $(\baseb, \mSet)$ such that $X$ contains exactly all finite representable sets of models closest to $\modelsof{\baseb} \cup \mSet$.  

We  also define the function $\mathcal{C}^+:  \powerset^*(\FRsets(\logsys)) \to \powerset(\mUni)$ as
\[\mathcal{C}^+(X) = \{ \modelsof{\mExp(\baseb, \mSet)} \mid  (\baseb, \mSet) \in \xi^+(X) \}.\]

\begin{lemma}\label{lem:app:modExpxi}
Let $\logsys$ be a \mexnm-compatible satisfaction system.  
If a model change operation $\mExp$ satisfies \textit{uniformity} then for all $X \in \powerset(\FRsets(\logsys))$: 

\begin{enumerate}[label= (\roman*), leftmargin=*]
    \item $\modelsof{\mExp(\baseb, \mSet)} = \modelsof{\mExp(\baseb', \mSet')}$
        for all $(\baseb, \mSet),  (\baseb', \mSet') \in \xi^+(X)$; and
    \item $\mathcal{C}^+(X)$ is a singleton,  if $\xi^+(X) \neq \emptyset$.
\end{enumerate}
\begin{proof}
Let $\mExp$ be a model change operation satisfying \textit{uniformity},  and $X \in \FRsets(\logsys)$, where $\logsys$ is a reception compatible satisfaction system.  

\begin{enumerate}[label= (\roman*), leftmargin=*]
	\item Let $(\baseb, \mSet),  (\baseb', \mSet') \in \xi^+(X)$.  Thus by definition of $\xi^+$, we have that 
	\begin{align*} 
        X   &= \FRsubs(\modelsof{\baseb} \cup \mSet, \logsys)\\ 
        {}  &= \FRsubs(\modelsof{\baseb'} \cup \mSet', \logsys).
    \end{align*}
Thus,  from \textit{uniformity}, we get 
\[ \modelsof{\mExp(\baseb, \mSet)}  =  \modelsof{\mExp(\baseb', \mSet')}.\]

\item Let us suppose that $\xi^+(X) \neq \emptyset$.  Then, there is some $(\baseb, \mSet) \in \xi^+(X)$.  Let us fix such a $(\baseb, \mSet)$.  
	By definition of $\xi^+$, we have that 
	\[ X = \FRsubs(\baseb \cup \mSet, \logsys)\]
	By definition of $\mathcal{C}^+$,  
\[ \modelsof{\mExp(\baseb,\mSet)} \in \mathcal{C}^+(X)\]	
Thus,  to show $\mathcal{C}^+(X)$ is a singleton,  we need to show that for all $Y \in \mathcal{C}^+(X)$,  $Y = \modelsof{\mExp(\baseb,\mSet)}$.   Let $Y \in \mathcal{C}^+(X)$.  By definition of $\mathcal{C}^+$, we have that 
\[ Y = \modelsof{\mExp(\baseb', \mSet')},  \mbox{ for some }  (\baseb', \mSet') \in \xi^+(X) \]
Thus,  as both pairs $(\baseb, \mSet),  (\baseb', \mSet') \in \xi^+(X)$,  we get from item (i) above that $\modelsof{\mExp(\baseb, \mSet)} = \modelsof{\mExp(\baseb', \mSet')}$.  Thus,  
$Y = \modelsof{\mExp(\baseb, \mSet)}$.  This concludes the proof. \qedhere  
\end{enumerate}
\end{proof}
\end{lemma}

\begin{proposition}\label{obs:_finiteMINFR}
Given a \mexnm-compatible satisfaction system $\logsys = (\llang, \mUni, \models)$.  
If a model change operation $\mExp$ satisfies \textit{success, persistence} and \textit{finite temperance},  then $\modelsof{\mExp(\baseb, \mSet)} \in \FRsups(\modelsof{\baseb} \cup \mSet, \logsys)$ for all \(\baseb \in \finitepwset(\logsys)\) and \(\mSet \subseteq \mUni\). 
\begin{proof}
Let us suppose for contradiction that there is a model change operation that satisfies  \textit{success, inclusion} and \textit{finite temperance},  
but $\modelsof{\mExp(\baseb, \mSet)} \not \in \FRsups(\modelsof{\baseb} \cup \mSet, \logsys)$,  for some finite base $\baseb$ and set of models $\mSet$.  
Let us fix such a base $\baseb$ and set $\mSet$.  

From \textit{success} and \textit{inclusion}, we have that
 \[ \modelsof{\baseb} \cup \mSet \subseteq \modelsof{\mExp(\baseb, \mSet)} \]
 
 From construction,  $\mExp(\baseb, \mSet)$ is a finite base,  which means 
 \begin{align}
 \modelsof{\mExp(\baseb, \mSet)} \in \FRsets(\logsys). \label{eq:mExpsys}
 \end{align}
  As $\logsys$ is eviction compatible,  $\FRsups(\modelsof{\baseb} \cup \mSet, \logsys) \neq \emptyset$.  
  Let 
  \[Y = \{ X \in \FRsets(\logsys) \mid  (\modelsof{\baseb} \cup \mSet) \subseteq X\}. \] 
  We have  $\modelsof{\mExp(\baseb, \mSet)} \not \in \FRsups(\modelsof{\baseb} \cup \mSet, \logsys)$ from hypothesis,  which means that either $\modelsof{\mExp(\baseb, \mSet)} \not \in \FRsets(\logsys)$ or 
  $\modelsof{\mExp(\baseb, \mSet)}$ is not $\subseteq$-minimal within 
  $Y$.  
  This fact taken together with \cref{eq:mExpsys} implies that 
  $\modelsof{\mExp(\baseb, \mSet)}$ is not $\subseteq$-minimal within $Y$.  
  Therefore,   there is some $\mSet' \in \FRsets(\logsys)$ such that 
  \begin{align}
  (\modelsof{\baseb} \cup \mSet) \subseteq \mSet'  \mbox { and } \mSet' \subset \modelsof{\mExp(\baseb, \mSet)} \label{eq:auxFRMin}
  \end{align}
Note that $\mSet \subseteq \mSet'$ and $\mSet \subseteq \modelsof{\mExp(\baseb,M)}$ which implies from \cref{eq:auxFRMin} above that 
$\mSet' \subset \modelsof{\mExp(\baseb, \mSet)} \cup \mSet$.  
  This implies from \textit{finite temperance} that $\mSet' \not \in \FRsets(\logsys)$, which is a contradiction.  
\end{proof}
\end{proposition}

\mFRExpRepr*

\begin{proof}
``$\Rightarrow$''
    Let \(\mExp_\FRsel\) be a maxichoice reception function
    over \(\logsys\) based on a \frselnm{}
    \(\FRsel\). 
Success follows directly from the construction of $\mExp$.  
For persistence,  note that 
    \(\modelsof{\mExp_\FRsel(\baseb, \mSet)} \in \FRsups(\modelsof{\baseb}
    \cup\mSet, \logsys)\),  which implies that $\modelsof{\baseb} \subseteq \modelsof{\mExp_{\FRsel}(\baseb, \mSet)}$.
 \begin{description}
 	\item[(vacuity)] Assume that $\mSet \subseteq \modelsof{\baseb}$.  Thus,  $\modelsof{\baseb} = \modelsof{\baseb} \cup \mSet$. 
	Thus,  $\FRsups(\modelsof{\baseb} \cup \mSet, \logsys) = \{ \modelsof{\baseb}\}$ which implies that 
	\[ \mExp(\baseb, \mSet) = \FRsel(\FRsups(\modelsof{\baseb} \cup \mSet, \logsys)) = \baseb',\]
	such that $\modelsof{\baseb'} = \modelsof{\baseb}$.  Thus,  
	\[\modelsof{\mExp(\baseb, \mSet)} = \modelsof{\baseb}\]
	
	\item[(finite temperance)] Suppose that
    \(\mSet' \not\in \modelsof{\baseb} \cup \mSet\) but \(\mSet' \in
    \modelsof{\mExp(\baseb, \mSet)}\),  then,  by construction, there is no
    \(\mSet \in \FRsets(\logsys)\) which contains \(\modelsof{\baseb} \cup
    \mSet\) and \(\mSet \subset \modelsof{\mExp_\FRsel(\baseb, \mSet)}\).
    
    \item[(uniformity)] 
    Let  $\FRsups(\modelsof{\baseb} \cup \mSet,  \logsys) =\FRsups(\modelsof{\baseb'} \cup \mSet',  \logsys)$.  By definition,  
    \begin{align*}
        \modelsof{\mExp(\baseb,\mSet)} &= \FRsel(\FRsups(\mathbb{Y}, \logsys))\\
        \text{and}\\
        \modelsof{\mExp(\baseb',\mSet')} &= \FRsel(\FRsups(\mathbb{Y}', \logsys)),
    \end{align*}
    where \(\mathbb{Y} = \modelsof{\baseb} \cup \mSet\) and \(\mathbb{Y'} = \modelsof{\baseb'} \cup \mSet'\).
    	
    	Thus,  as $\FRsups(\modelsof{\baseb} \cup \mSet, \logsys) = \FRsups(\modelsof{\baseb'} \cup \mSet',  \logsys)$,  we get that \[\modelsof{\mExp(\baseb,\mSet)} = \modelsof{\mExp(\baseb',\mSet')}.\]  
 \end{description}
    
    Hence, every maxichoice reception function based on
    a \frselnm{} satisfies all
     postulates stated above.

``$\Leftarrow$'' 
Let \(\mExp : \finitepwset(\llang) \times \powerset(\mUni) \to
    \finitepwset(\llang)\) be a function satisfying the postulates stated. 
	As $\mExp$ satisfies uniformity,  we known from \cref{lem:app:modExpxi} that $\mathcal{C}^+(X)$ is a singleton for every $X \in \powerset(\FRsets(\logsys))$.  Thus,  we can construct the function 
	$\FRsel: \powerset^*(\FRsets(\logsys)) \to \FRsets(\logsys)$ such that
		\begin{align*}
            \FRsel(X) =
            \begin{dcases}
                Z \text{ s.t. } \mathcal{C}^+(X) = \{Z\} & \text{if } \xi^+(X) \neq \emptyset,\\
                Y \mbox{ s.t.  } Y \in X & \text{otherwise.}
            \end{dcases}
		\end{align*}
	We will prove that: (i) $\FRsel$ is indeed a selection function,  and (ii) 
	that $\modelsof{\mExp(\baseb, \mSet)} = \FRsel(\FRsups(\modelsof{\baseb} \cup \mSet, \logsys))$.
	
\begin{enumerate}[label= (\roman*), leftmargin=*]
	\item \textit{$\mathbf{\FRsel}$ is indeed a selection function.  }
        Let $X \in \nepwset(\FRsets(\logsys))$.
		We only need to show that $\FRsel(X) \in X$.  
		The case that $\xi^+(X) = \emptyset$  is trivial,  as $\FRsel$ chooses an arbitrary $Y \in X$ (by the axiom of choice).  Let us focus on the case $\xi^+(X) \neq \emptyset$.  
		From above,  we have that $\FRsel(X) = Z$,  where $\mathcal{C}^+(X) = \{Z \}$.  
		By definition of $\mathcal{C}^+$,  we have that
		 there is a pair  $(\baseb, \mSet) \in \xi^+(X)$ such that
		\[ Z = \modelsof{\mExp(\baseb,\mSet)}. \]
		Let us fix such a $(\baseb, \mSet) \in \xi^+(X)$.  
		Thus,  by definition of $\xi^+$, we get
		\[ X = \FRsups(\modelsof{\baseb} \cup \mSet, \logsys).\]
		
From \cref{obs:_finiteMINFR}, we get that $\modelsof{\mExp(\baseb,\mSet)} \in  \FRsups(\modelsof{\baseb} \cup \mSet, \logsys)$. Thus,  from the identities above we get that $Z \in X$,  which means $\FRsel(X) \in X$. 
	\item $\modelsof{\mExp(\baseb, \mSet)} = \FRsel(\FRsups(\modelsof{\baseb} \cup \mSet, \logsys))$.  
		Let $X = \FRsups(\modelsof{\baseb}\cup \mSet, \logsys)$.  
        We know that \(X \neq \emptyset\) due to \mexnm-compatibility.
			By definition of $\xi^+$, we get that 
			\[(\baseb, \mSet) \in \xi^+(X). \]
			By construction,  we have that $\FRsel(X) = Z$ such that  	$\mathcal{C}^+(X) = \{ Z\}$,  which implies  from definition of $\mathcal{C}^+$ that $Z = \modelsof{\mExp(\baseb', \mSet')}$,  for some $(\baseb', \mSet') \in \xi^+(X)$.  
			Therefore,  as 	$(\baseb, \mSet) \in \xi^+(X)$, we get 
			from \cref{lem:app:modExpxi},  that for all  $(\baseb', \mSet') \in \xi^+(X),  $
		\[  \modelsof{\mExp(\baseb, \mSet)} = \modelsof{\mExp(\baseb', \mSet')}. \]
	Thus,  
$\modelsof{\mExp(\baseb, \mSet)} = Z$.  
Thus as $\FRsel(X) = Z$ and $X = \FRsups(\modelsof{\baseb}\cup \mSet, \logsys)$, we have that 
\begin{align*} 
    &\modelsof{\mExp(\baseb, \mSet)} =\\
    &\qquad\FRsel(\FRsups(\modelsof{\baseb} \cup \mSet, \logsys)).
\end{align*}%
\end{enumerate}

\end{proof}

\vacuityExp*
\begin{proof}
Assume that $\mExp$ satisfies finite temperance and persistence,  and that $\mSet \subseteq \modelsof{\baseb}$.  
This means that $\modelsof{\baseb}$ is the closest finite representable superset of $\baseb$ containing $\mSet$.  
From persistence,  $\modelsof{\baseb} \subseteq \modelsof{\mExp(\baseb, \mSet)}$.  Thus,  from finite temperance,  
we get that $\modelsof{\mExp(\baseb, \mSet)} = \modelsof{\baseb}$.  
\end{proof}

\UniqueFR*

\begin{proof}
    We only prove the result for \(\FRsups\), as the case for \(\FRsubs\) is
    analogous.
    Assume that \(\mSet_1, \mSet_2 \in \FRsups(\mSet, \logsys)\) with  \(\mSet_1
    \neq \mSet_2\).
    Since \(\mSet_1\) and \(\mSet_2\) are finitely representable
    in \(\llang\), there are two
    finite bases \(\baseb_1\), \(\baseb_2\) such that \(\modelsof{\baseb_1}
    = \mSet_1\) and \(\modelsof{\baseb_2} = \mSet_2\). 

    First, let \(\mSet' = \modelsof{\baseb_1 \cup \baseb_2}\) and
    \(\modelm \in \mSet\). We know that \(\modelm \in
    \mSet_1 \cap \mSet_2\) because both \(\mSet_1\) and \(\mSet_2\) are
    supersets of \(\mSet\). From the RMBP, it
    holds that \(\modelm \in \modelsof{\baseb_1 \cup \baseb_2}\).
    Since the choice of \(\modelm\) was
    arbitrary, we can conclude that \(\mSet \subseteq \mSet'\).

    Now, let \(\modelm \in \mSet'\). Due to the RMBP we have that \(\modelm \in
    \modelsof{\baseb_1} = \mSet_1\) and \(\modelm \in \modelsof{\baseb_2}
    = \mSet_2\). That is, \(\mSet' \subseteq \mSet_1 \cap \mSet_2\).
    Therefore, \(\mSet'\) is a finitely representable (just take \(\baseb_1 \cup
    \baseb_2\) as the base) superset of \(\mSet\). However, since we assume that 
    \(\mSet_1, \mSet_2 \in \FRsups(\mSet, \logsys)\), by minimality we get that 
    \(\mSet_1 = \mSet' \subseteq \mSet_1 \cap \mSet_2\) which implies
    \(\mSet_1 = \mSet_2\), a contradiction. Hence, there can be
    at most one set of models in \(\FRsups(\mSet, \logsys)\).
\end{proof}

%
%

\section{Proofs for \Cref{sec:uniq}}

First, we prove in our claim about \mconnm- and \mexnm-compatibility of \(\logsys_q\) from \cref{ex:rational} with \cref{prop:rational}.

\begin{proposition}%
\label{prop:rational}
    Let \(\logsys_{q} = (\llang_q, \mUni_q, \models_q)\) be such that
    \(\llang_q = \{[x, y] \mid x, y \in \mathbb{Q} \text{ and } x \leq y\}\),
    \(\mUni_q = \mathbb{Q}\) and
    \(Q \models_q \baseb\) (with \(Q \subseteq \mathbb{Q}\)) iff for all \(z \in Q\), \(x \leq z \leq y\) for every \([x, y] \in \baseb\).
\end{proposition}

\begin{proof}

    We will show that this system is not \mconnm-compatible.
    Consider the base \(\{[0, 1]\}\) and the set of models \(\{1\}\). 
    Since the language only admits closed intervals and by definition of \(\models_q\), any finite base in \(\llang_q\) is either inconsistent or equivalent to a single continuous interval.
    Therefore, for any \(\baseb' \in \finitepwset(\llang_q)\) that does not include \(\{1\}\) there will always be a finite base that has more models.
    More precisely, let \(\mathopen[x', y'\mathclose]\) be the interval corresponding to a candidate finite base \(\baseb'\).
    We can assume without loss of generality that \(y' < 1\) and we know that there are infinitively many rational numbers between \(y'\) and \(1\). 
    Thus, we can always extend the interval to a new rational, capturing more models than before, without losing finite representability or including \(1\) in the models of the base.
    Therefore \(\FRsubs(\left[0, 1\right), \logsys_q) = \emptyset\), that is, \(\logsys_q\) is not \mconnm-compatible.

    Now, we will prove that \(\logsys_q\) is not \mexnm-compatible.
    Consider the base \(\{\mathopen[0.5,1\mathclose]\}\) and the set of models \(\mathopen(0, 1\mathclose]\).
    Using the same argument as before, we can conclude that \(\FRsups(\mSet, \logsys_q)\) corresponds to either the smallest closed interval containing \(\left(0, 1\right]\).
    Since \(\{[0, 1]\}\) is finitely representable, any candidate must be equivalent to a closed interval \([x', y']\) such that \(0 < x' < y' = 1\).
    Otherwise, either it would not be a superset of \(\left(0, 1\right]\), or would include too many models, losing minimality.
    However, for any \(x' \in \mathbb{Q}\) with \(0 < x' < 1\) there is a \(x^{\prime\prime}\) with \(0 < x^{\prime\prime} < x^\prime\).
    This means that we can always find a candidate finite base that has fewer models. 
    Therefore, \(\FRsups(\left(0, 1\right], \logsys_q) = \emptyset\), that is, \(\logsys_q\) is not \mexnm-compatible.
\end{proof}

\begin{proposition}\label{thirdEvcCond}
    Let \(\logsys = (\llang, \mUni, \models)\) and \(\mSet \subseteq \mUni\).
    Then, there are \(\baseb \in \finitepwset(\llang)\) and \(\mSet' \subseteq \mUni\) such that \(\modelsof{\baseb} \setminus \mSet' = \mSet\), iff there is a \(\mSet^{\prime\prime} \in \FRsets(\logsys)\) with \(\mSet \subseteq \mSet^{\prime\prime}\).
\end{proposition}

\begin{proof}
    \(\Rightarrow\): If we suppose that there are \(\baseb \in \finitepwset(\llang)\) and \(\mSet' \subseteq \mUni\) such that \(\modelsof{\baseb} \setminus \mSet' = \mSet\), then we can take \(\mSet^{\prime\prime} = \modelsof{\baseb}\).
    
    \(\Leftarrow\): Assuming that there is a \(\mSet^{\prime\prime} \in \FRsets(\logsys)\) with \(\mSet \subseteq \mSet^{\prime\prime}\), we can take \(\baseb \in \finitepwset(\llang)\) such that \(\modelsof{\baseb} = \mSet^{\prime\prime}\) and \(\mSet' = (\mUni \setminus \mSet)\). Then \(\modelsof{\baseb} \setminus \mSet' = \mSet^{\prime\prime} \setminus \mUni \setminus \mSet\), and as \(\mSet \subseteq \mSet^{\prime\prime} \subseteq \mUni\), we get \(\modelsof{\baseb} \setminus \mSet' = \mSet\).
\end{proof}

\IncFRsubs*
\begin{proof} 
	We split the statement of the theorem into
	the following two claims, which directly imply the theorem. 
	\begin{claim}
A satisfaction system \( \logsys = (\llang, \mUni, \models)\)  is \mconnm-compatible iff for every \(\mSet \subseteq \mUni\) either (i) \(\mSet \in \FRsets(\logsys)\), (ii) \(\mSet\) has an immediate predecessor in
\((\FRsets(\logsys) \cup \{\mSet\}, \subset)\), or (iii) there is no \(\mSet' \in \FRsets(\Lambda)\) with \(\mSet \subseteq \mSet'\).
	\end{claim}
\begin{proof} 
	
    \(\Rightarrow\): Suppose that \(\logsys\) is \mconnm-compatible, that is, \(\FRsups(\modelsof{\baseb} \setminus \mSet, \logsys) \neq \emptyset\) for all \(\baseb \in \finitepwset(\llang)\) and \(\mSet \subseteq \mUni\).

    Let \(\mSet_1 \subseteq \mUni\). If \(\mSet_1 \in \FRsets(\logsys)\) then the \lcnamecref{incFRsubs} holds trivially.

    Now, we consider two cases \(\FRsubs(\mSet_1, \logsys) \neq \emptyset\) and \(\FRsubs(\mSet_1, \logsys) = \emptyset\).
	
	In the first case, we know there is a \(\mSet_2 \in \FRsubs(\mSet_1, \logsys)\). 
    We will show that \(\mSet_2\) is an immediate predecessor of \(\mSet_1\).
    Since \(\mSet_2 \in \FRsubs(\mSet_1, \logsys)\), \(\mSet_2 \subseteq \mSet_1\) and by \cref{def:FRsubs} there is no \(\mSet_2^{\prime} \in \FRsets(\logsys)\) such that \(\mSet_2 \subset \mSet_2^{\prime} \subset \mSet_1\).
    Consequently, \(\mSet_2\) is an immediate predecessor of \(\mSet_1\) in
    \((\FRsets(\logsys) \cup \{\mSet_1\}, \subset)\).

    In the second case, due to \mconnm-compatibility, we know that there is no \(\baseb \in \finitepwset(\llang)\) and \(\mSet_3 \subseteq \mUni\) such that \(\mSet_1 = \modelsof{\baseb} \setminus \mSet_3\). 
    Therefore, we can use \cref{thirdEvcCond} to conclude that there is no \(\mSet' \in \FRsets(\Lambda)\) with \(\mSet_1 \subseteq \mSet'\).
	
	\(\Leftarrow\): Assume that for all \(\mSet \subseteq \mUni\), \(\mSet \in \FRsets(\logsys)\), \(\mSet\) has an immediate predecessor in \((\FRsets(\logsys) \cup \{\mSet\}, \subset)\), or there is no \(\mSet' \in \FRsets(\Lambda)\) with \(\mSet \subseteq \mSet'\).
    Let \(\mSet_1 \subseteq \mUni\). 
    We consider following cases.

    \begin{enumerate}[label= (\roman*), leftmargin=*]
        \item \(\mSet_1 \in \FRsets(\logsys)\): by \cref{def:FRsubs} we have that \(\FRsubs(\mSet_1, \logsys) = \{\mSet_1\} \neq \emptyset\).
        \item \(\mSet_1\) has an immediate predecessor in \((\FRsets(\logsys) \cup \{\mSet_1\}, \subset)\): then there is a \(\mSet_2 \in \FRsets(\logsys)\) such that \(\mSet_2 \subset \mSet_1\) and there is no \(\mSet_2^{\prime} \in \FRsets(\logsys)\) such that \(\mSet_2 \subset \mSet_2^{\prime} \subset \mSet_1\). In other words, \(\mSet_2 \in \FRsubs(\mSet_1, \logsys)\).
        \item There is no \(\mSet' \in \FRsets(\Lambda)\) with \(\mSet_1 \subseteq \mSet'\): then, we know from \cref{thirdEvcCond} that there is no \(\baseb \in \finitepwset(\llang)\) and \(\mSet^{\prime\prime} \subseteq \mUni\) such that \(\mSet_1 = \modelsof{\baseb} \setminus \mSet^{\prime\prime}\).
    \end{enumerate}
    Hence, if there are \(\baseb \in \finitepwset(\logsys)\) and \(\mSet \in \mUni\) such that \(\modelsof{\baseb} \setminus \mSet = \mSet_1\), then \(\FRsubs(\mSet_1, \logsys) \neq \emptyset\).
    Since the choice of \(\mSet_1\) was arbitrary, we can conclude that \(\logsys\) is \mconnm-compatible.
\end{proof}

\begin{proposition}\label{thirdRcpCond}
    Let \(\logsys = (\llang, \mUni, \models)\) and \(\mSet \subseteq \mUni\).
    There are \(\baseb \in \finitepwset(\llang)\) and \(\mSet' \subseteq \mUni\) such that \(\modelsof{\baseb} \cup \mSet' = \mSet\), iff there is a \(\mSet^{\prime\prime} \in \FRsets(\logsys)\) with \(\mSet^{\prime\prime} \subseteq \mSet\).
\end{proposition}

\begin{proof}
    \(\Rightarrow\): If we suppose that there are \(\baseb \in \finitepwset(\llang)\) and \(\mSet' \subseteq \mUni\) such that \(\modelsof{\baseb} \cup \mSet' = \mSet\), then we can take \(\mSet^{\prime\prime} = \modelsof{\baseb}\).
    
    \(\Leftarrow\): Assuming that there is a \(\mSet^{\prime\prime} \in \FRsets(\logsys)\) with \(\mSet^{\prime\prime} \subseteq \mSet\), we can take \(\baseb \in \finitepwset(\llang)\) such that \(\modelsof{\baseb} = \mSet^{\prime\prime}\) and \(\mSet' = \mSet\).
    Then \(\modelsof{\baseb} \cup \mSet' = \mSet^{\prime\prime} \cup \mSet\), and as \(\mSet^{\prime\prime} \subseteq \mSet\), we get \(\modelsof{\baseb} \cup \mSet' = \mSet\).
\end{proof}

\begin{claim}
A satisfaction system \( \logsys = (\llang, \mUni, \models)\) is \mexnm-compatible iff for every \(\mSet \subseteq \mUni\) either (i) \(\mSet \in \FRsets(\logsys)\), (ii) \(\mSet\) has an immediate successor in
\((\FRsets(\logsys) \cup \{\mSet\}, \subset)\), or (iii) there is no \(\mSet' \in \FRsets(\Lambda)\) with \(\mSet' \subseteq \mSet\).
\end{claim}
\begin{proof} 
    
    \(\Rightarrow\): Suppose that \(\logsys\) is \mexnm-compatible, that is, \(\FRsubs(\modelsof{\baseb} \cup \mSet, \logsys) \neq \emptyset\) for all \(\baseb \in \finitepwset(\llang)\) and \(\mSet \subseteq \mUni\).

    Let \(\mSet_1 \subseteq \mUni\). If \(\mSet_1 \in \FRsets(\logsys)\) then the theorem 
     holds trivially.

    Now, we consider two cases \(\FRsups(\mSet_1, \logsys) \neq \emptyset\) and \(\FRsups(\mSet_1, \logsys) = \emptyset\).
	
	In the first case, we know there is a \(\mSet_2 \in \FRsups(\mSet_1, \logsys)\). 
    We will show that \(\mSet_2\) is an immediate successor of \(\mSet_1\).
    Since \(\mSet_2 \in \FRsups(\mSet_1, \logsys)\), \(\mSet_1 \subseteq \mSet_2\) and by \cref{def:FRsups} there is no \(\mSet_2^{\prime} \in \FRsets(\logsys)\) such that \(\mSet_1 \subset \mSet_2^{\prime} \subset \mSet_2\).
    Consequently, \(\mSet_2\) is an immediate successor of \(\mSet_1\) in
    \((\FRsets(\logsys) \cup \{\mSet_1\}, \subset)\).

    In the second case, due to \mexnm-compatibility, we know that there is no \(\baseb \in \finitepwset(\llang)\) and \(\mSet_3 \subseteq \mUni\) such that \(\mSet_1 = \modelsof{\baseb} \cup \mSet_3\). 
    Therefore, we can use \cref{thirdRcpCond} to conclude that there is no \(\mSet' \in \FRsets(\Lambda)\) with \(\mSet' \subseteq \mSet_1\).
	
	\(\Leftarrow\): Assume that for all \(\mSet \subseteq \mUni\), \(\mSet \in \FRsets(\logsys)\), \(\mSet\) has an immediate successor in \((\FRsets(\logsys) \cup \{\mSet\}, \subset)\), or there is no \(\mSet' \in \FRsets(\Lambda)\) with \(\mSet' \subseteq \mSet\).
    Let \(\mSet_1 \subseteq \mUni\). 
    We consider following cases.

    \begin{enumerate}[label= (\roman*), leftmargin=*]
        \item \(\mSet_1 \in \FRsets(\logsys)\): by \cref{def:FRsups} we have that \(\FRsups(\mSet_1, \logsys) = \{\mSet_1\} \neq \emptyset\).
        \item \(\mSet_1\) has an immediate successor in the poset \((\FRsets(\logsys) \cup \{\mSet_1\}, \subset)\): then there is a \(\mSet_2 \in \FRsets(\logsys)\) such that \(\mSet_1 \subset \mSet_2\) and there is no \(\mSet_2^{\prime} \in \FRsets(\logsys)\) such that \(\mSet_1 \subset \mSet_2^{\prime} \subset \mSet_2\). In other words, \(\mSet_2 \in \FRsups(\mSet_1, \logsys)\).
        \item There is no \(\mSet' \in \FRsets(\Lambda)\) with \(\mSet' \subseteq \mSet_1\): then, we know from \cref{thirdRcpCond} that there is no \(\baseb \in \finitepwset(\llang)\) and \(\mSet^{\prime\prime} \subseteq \mUni\) such that \(\mSet_1 = \modelsof{\baseb} \cup \mSet^{\prime\prime}\).
    \end{enumerate}
    Hence, if there are \(\baseb \in \finitepwset(\logsys)\) and \(\mSet \in \mUni\) such that \(\modelsof{\baseb} \cup \mSet = \mSet_1\), then \(\FRsups(\mSet_1, \logsys) \neq \emptyset\).
    By the arbitrariety of \(\mSet_1\) we can conclude that \(\logsys\) is \mexnm-compatible.
\end{proof}
\end{proof}

\EvcRcpSuff*

\begin{proof}
    Since \(\FRsets(\logsys)\) is finite, the existence of an immediate
    predecessor is guaranteed for all \(\emptyset \neq \mSet \subseteq \mUni\)
    and so is ensured the existence of an immediate successor for all \(\mSet
    \subset \mUni\). Therefore, this result is a direct consequence of \cref{incFRsubs} (Item 1)
    for the first point and of \cref{incFRsubs} (Item 2)
    for the second point.
\end{proof}

\section{Proofs for \Cref{sec:usecases}}

\subsection{Proofs for \cref{use:prop}}

\propcompat*

\begin{proof}
    Since we need only to consider finitely many symbols, there are finitely many possible valuations.
    If there are \(n\) propositional atoms, there are at most \(2^n\) distinct models, meaning that there are at most 
    \(2^{m}\) distinct sets of valuations where \(m = 2^n\).
    Consequently, \(\FRsets(\logsysprop)\) is finite.
    Additionally, since both the empty set and the set of all valuations are representable in this satisfaction system, we obtain as a consequence of \cref{evcRcpSuff} that \(\logsysprop\) is both \mconnm- and \mexnm-compatible.
\end{proof}

\begin{proposition}\label{evcProp}
    Let \(\logsysprop\) be the satisfaction system with the entailment relation given by the standard semantics of propositional logic with finite signature.
    The function \(\evcprop\) defined next is a maxichoice \mconnm{} on \(\logsysprop\).
    \begin{align*}
        \evcprop(\baseb,\mSet) = \bigvee_{v \in \modelsof{\baseb} \setminus \mSet} \left(\bigwedge_{v(a) = \vtrue} a \wedge \bigwedge_{v(a) = \vfalse} \neg{a} \right)
    \end{align*}
\end{proposition}

\begin{proof}
    We will use \cref{mFRCon_r} to prove this result, by showing that \(\evcprop\) satisfies each of the postulates stated.
    Recall that each model is a valuation over a finite number of propositional atoms, and therefore, the set of all models is finite.
    \begin{description}[leftmargin=*]
        \item[(success)] {%
                Let \(v \in \mSet\). 
                Clearly, \(v \not\in \modelsof{\baseb} \setminus \mSet\).
                We know that \(v\) does not satisfy any of the disjuncts that compose \(\evcprop\), as each is satisfied by exactly one valuation.
                It follows from the standard semantics of proposition logic with finite signature that \(v \not\in \modelsof{\evcprop(\baseb, \mSet)}\).
                As we only assumed that \(v \in \mSet\), we can conclude that \(\mSet \cap \modelsof{\evcprop(\baseb, \mSet)} = \emptyset\).
            }
        \item[(inclusion)]{%
                Let \(v \not\in \modelsof{\baseb}\).
                Consequently, \(v\) does not satisfy any of the disjuncts that compose \(\evcprop\), as each is satisfied by exactly one valuation.
                It follows from the standard semantics of proposition logic with finite signature that \(v \not\in \modelsof{\evcprop(\baseb, \mSet)}\).
                Since \(v\) was arbitrarily chosen, we obtain \(\modelsof{\evcprop(\baseb, \mSet)} \subseteq \modelsof{\baseb}\).
            } 
        \item[(vacuity)]{%
                If $\mSet \cap \modelsof{\baseb} = \emptyset$ then 
                \begin{align*}
                    \evcprop(\baseb,\mSet) = \bigvee_{v \in \modelsof{\baseb}} \left(\bigwedge_{v(a) = \vtrue} a \wedge \bigwedge_{v(a) = \vfalse} \neg{a} \right).
                \end{align*}
                Since each disjunct is associated to exactly one model, every model of \(\baseb\) will also be a model of \(\evcprop(\baseb, \mSet)\), and exactly those, i.e., \(\modelsof{\evcprop(\baseb, \mSet)} = \modelsof{\baseb}\).
            }
        \item[(finite retainment)]{
                Each disjunct of \(\evcprop(\modelsof{\baseb}, \mSet)\) is associated to exactly one model in \(\modelsof{\baseb} \setminus \mSet\), hence \(\modelsof{\evcprop(\baseb, \mSet)} = \modelsof{\baseb} \setminus \mSet\).
                Therefore, there is no \(\mSet^\prime \in \FRsets(\logsysprop)\) such that \(\modelsof{\evcprop(\baseb, \mSet)} \subset \mSet^\prime \subseteq \modelsof{\baseb} \setminus \mSet\).
        }
        \item[(uniformity)] {%
                In \(\logsysprop\) every set of models is finitely representable, thus, $\FRsubs(\modelsof{\baseb} \setminus \mSet, \logsysprop) = \modelsof{\baseb} \setminus \mSet$.
                Therefore, if $\FRsubs(\modelsof{\baseb} \setminus \mSet, \logsysprop) = \FRsubs(\modelsof{\baseb'} \setminus \mSet', \logsysprop)$ then \(\modelsof{\baseb} \setminus \mSet = \modelsof{\baseb'} \setminus \mSet'\). 
                In this case, we have that \(\evcprop(\baseb, \mSet) = \evcprop(\baseb', \mSet')\) which implies $\modelsof{\evcprop(\baseb, \mSet)} = \modelsof{\evcprop(\baseb',  \mSet')}$.
            }
    \end{description}
    Since \(\evcprop\) satisfies all the postulates from \cref{mFRCon_r}, it follows that it is a maxichoice \mconnm{} function over \(\logsysprop\).
\end{proof}

\begin{proposition}\label{rcpProp}
    Let \(\logsysprop\) be the satisfaction system with the entailment relation given by the standard semantics of propositional logic with finite signature.
    The function \(\rcpprop\) defined next is a maxichoice \mexnm{} on \(\logsysprop\).
    \begin{align*}
        \rcpprop(\baseb,\mSet) = \bigvee_{v \in \modelsof{\baseb} \cup \mSet} \left(\bigwedge_{v(a) = \vtrue} a \wedge \bigwedge_{v(a) = \vfalse} \neg{a} \right)
    \end{align*}
\end{proposition}

\begin{proof}
    We will use \cref{mFRExp_r} to prove this result by showing that \(\rcpprop\) satisfies each of the postulates stated.
    Recall that each model is a valuation over a finite number of propositional atoms, and therefore, the set of all models is finite.
    \begin{description}[leftmargin=*]
        \item[(success)] {%
                Let \(v \in \mSet\). 
                Clearly, \(v \in \modelsof{\baseb} \cup \mSet\).
                Consequently, \(v\) satisfies one of the disjuncts that compose \(\rcpprop\), as each is satisfied by exactly one valuation.
                It follows from the standard semantics of proposition logic with finite signature that \(v \in \modelsof{\rcpprop(\baseb, \mSet)}\).
                As we only assumed that \(v \in \mSet\), we can conclude that \(\mSet \subseteq \modelsof{\rcpprop(\baseb, \mSet)}\).
            }
        \item[(persistence)]{%
                Let \(v \in \modelsof{\baseb}\).
                We know that \(v\) satisfies one of the disjuncts that compose \(\rcpprop\), as each is satisfied by exactly one valuation.
                It follows from the standard semantics of proposition logic with finite signature that \(v \in \modelsof{\rcpprop(\baseb, \mSet)}\).
                Since \(v\) was arbitrarily chosen, we obtain \(\modelsof{\baseb} \subseteq \modelsof{\rcpprop(\baseb, \mSet)}\).
            } 
        \item[(vacuity)]{%
                If $\mSet \subseteq \modelsof{\baseb} = \emptyset$ then 
                \begin{align*}
                    \rcpprop(\baseb,\mSet) = \bigvee_{v \in \modelsof{\baseb}} \left(\bigwedge_{v(a) = \vtrue} a \wedge \bigwedge_{v(a) = \vfalse} \neg{a} \right).
                \end{align*}
                Since each disjunct is associated to exactly one model, only models of \(\baseb\) will be a models of \(\rcpprop(\baseb, \mSet)\), that is, \(\modelsof{\rcpprop(\baseb, \mSet)} = \modelsof{\baseb}\).
            }
        \item[(finite temperance)]{
                Each disjunct of \(\rcpprop(\modelsof{\baseb}, \mSet)\) is associated to exactly one model in \(\modelsof{\baseb} \cup \mSet\), hence \(\modelsof{\rcpprop(\baseb, \mSet)} = \modelsof{\baseb} \cup \mSet\).
                Therefore, there is no \(\mSet^\prime \in \FRsets(\logsysprop)\) such that \(\modelsof{\baseb} \cup \mSet \subset \mSet^\prime \subseteq \modelsof{\rcpprop(\baseb, \mSet)}\).
        }
        \item[(uniformity)] {%
                In \(\logsysprop\) every set of models is finitely representable, thus, $\FRsups(\modelsof{\baseb} \cup \mSet, \logsysprop) = \modelsof{\baseb} \cup \mSet$.
                Therefore, if $\FRsups(\modelsof{\baseb} \cup \mSet, \logsysprop) = \FRsubs(\modelsof{\baseb'} \cup \mSet', \logsysprop)$ then \(\modelsof{\baseb} \cup \mSet = \modelsof{\baseb'} \cup \mSet'\). 
                In this case, we have that \(\rcpprop(\baseb, \mSet) = \rcpprop(\baseb', \mSet')\) which implies $\modelsof{\rcpprop(\baseb, \mSet)} = \modelsof{\mCon(\rcpprop',  \mSet')}$.
            }
    \end{description}
    Since \(\rcpprop\) satisfies all the postulates from \cref{mFRExp_r}, it follows that it is a maxichoice \mexnm{} function over \(\logsysprop\).
\end{proof}

\rPropOps*

\begin{proof}
    Direct consequence of \cref{evcProp,rcpProp}.
\end{proof}

\rhorncompat*

\begin{proof}
As for classical propositional logics,  we have that \(\FRsets(\logsyshorn)\) is finite.  
Observe that $\modelsof{\{a \to a\}} = \mUni$,  as $a \to a$ is tautological.  
Moreover,  the set $\modelsof{\{\perp\}} = \emptyset$.  Thus,  
both $\emptyset$ and $\mUni$ are finitely representable.  Therefore,  according to \cref{evcRcpSuff},  \(\logsyshorn\) is both eviction and reception compatible. 
\end{proof}

 \subsection{Proofs for \Cref{use:3vl}}

 \rthreevlcompat*

\begin{proof}[Sketch]
    As in the propositional case, \(\mUni_3\) is finite and \(\mUni_3\) are finitely representable in both systems. However, \(\emptyset\) is finitely representable in \(\logsyskleene\) but not in \(\logsyspriest\).
    Hence, the \lcnamecref{3vlcompat} is a consequence of \cref{evcRcpSuff}.
\end{proof}

\begin{proof}
    In both systems,  we have exactly the same set of models which is finite,  precisely we have  $3^{|\propatoms|}$ models,   where $\propatoms$ is the set of propositional symbols (which is assumed to be finite).
    Thus,  we have $2^m$ classes of equivalences of formulae,  where $m = 3^{|\propatoms|}$.  Thus,  for every $K \subseteq \llangprop$, there is a finite base \(\baseb \in \finitepwset(\llangprop)\) such that \(\modelsof{K} = \modelsof{\baseb}\).
    Observe that in both systems $\modelsof{\emptyset} = \mUniKleene$,  
    which means that $\mUniKleene$ is finitely representable in both $\logsyskleene$ and $\logsyspriest$.
    Also, $\modelsof{\llang} = \emptyset$,  in \(\logsyskleene\).
    Thus, as every set of formulae has  a finite base,  we get that $\llangprop$ also has a finite base in \(\logsyskleene\).
    However, the model that assigns \(\vunk\) to every propositional formula will satisfy any base according to \(\models_{P3}\).
    Thus, $\emptyset$ is finitely representable in $\logsyskleene$ but not in $\logsyspriest$.
    Therefore, it follows directly from \cref{evcRcpSuff} that both systems are and reception-compatible but \(\logsyskleene\) is \mconnm-compatible, while \(\logsyspriest\) is not.
\end{proof}

\subsection{Proofs for \cref{use:fuzzy}}

\begin{definition}%
    \label{def:fgodelss}
    Let \(\theta \in {\left(0, 1\right]}\). The satisfaction system of the propositional Gödel logic, in symbols, \(\logsysfuzzy\) is defined as \(\logsysfuzzy = (\llangfuzzy, \mUnifuzzy, \modelsfuzzy)\) in which
    \begin{itemize}
        \item \(\llangfuzzy\) consists of propositional formulas defined over a non-empty finite set of propositional atoms \(\propatoms\) and the connectives \(\land\), \(\lor\), \(\neg\), and \(\rightarrow\);
        \item{%
               \(\mUnifuzzy\) is the set of all functions \(v : \llang \to [0, 1]\) respecting the standard Gödel semantics for the boolean connectives given below
                \begin{align*}
                    v(\neg{\varphi}) &= \begin{cases}
                                            1 & \text{if } v(\varphi) = 0,\\
                                            0 & \text{otherwise;}
                                        \end{cases}\\
                    v(\varphi \land \psi) &= \min(v(\varphi), v(\psi));\\
                    v(\varphi \lor \psi) &= \max(v(\varphi), v(\psi));\\
                    v(\varphi \rightarrow \psi) &= \begin{cases}
                                                    1 & \text{if } v(\varphi) \leq v(\psi),\\
                                                    v(\psi) & \text{otherwise; and}
                                                \end{cases}
                \end{align*}
        }
        \item \(v \modelsfuzzy B\) iff \(v(\bigwedge_{\varphi \in B \cup \{(\neg{a} \lor a)\}} \varphi) \geq \theta\), with some \(a \in \propatoms\).
    \end{itemize}
\end{definition}

Henceforth, given \(v \in \mUnifuzzy\) and \(\varphi \in \llangfuzzy\), we will abuse the notation and write \(v \modelsfuzzy \varphi\) as a shorthand for \(v \modelsfuzzy \{\varphi\}\).

\begin{definition}%
\label{def:thetaext}
    Let \(\theta \in \left(0,1\right]\), \(\propatoms\) be a non-empty finite set of propositional atoms, \(\llangfuzzy\) defined over \(\propatoms\) and \(v : \llangfuzzy \to [0, 1]\).
    Also let \(s_\theta \not\in \propatoms\).
    We define the \(\theta\)-extension of \(v\) as:
    \(v^\ast : \llang \cup \{s_\theta\} \to [0, 1]\) defined as
    \begin{align*}
        v^\ast(\varphi) =
        \begin{cases}
            v(\varphi) & \text{if } \varphi \in \llangfuzzy,\\
            \theta & \text{if } \varphi \text{ is } s_\theta.\\
        \end{cases}
    \end{align*}
\end{definition}

\begin{definition}%
\label{def:inducedpreorder}
    Let \(\theta \in \left(0,1\right]\), \(\propatoms\) be a non-empty finite set of propositional atoms, \(\llangfuzzy\) defined over \(\propatoms\) and \(v : \llangfuzzy \to [0, 1]\).
    From the \(\theta\)-extension of \(v\) we define the following total preorders\footnote{A preorder is a binary relation that is reflexive and transitive.}

    \begin{itemize}
        \item \(\preceq_v \subseteq {(\llangfuzzy \cup \{s_\theta\}) \times (\llangfuzzy \cup \{s_\theta\})}\) such that \(\varphi \preceq_v \psi\) iff \(v^\ast(\varphi) \leq v^\ast(\psi)\); and
        \item \(\preceq^\prime_v \subseteq {(\propatoms \cup \{s_\theta\}) \times (\propatoms \cup \{s_\theta\}) }\) such that \(\varphi \preceq^\prime_v \psi\) iff \(v^\ast(\varphi) \leq v^\ast(\psi)\).
    \end{itemize}
\end{definition}

\begin{lemma}%
\label{preorderdep}
    Let \(v, w \in \mUnifuzzy\) with \(\preceq_v = \preceq_w\), then, for all \(\varphi \in \llangfuzzy\), \(v \modelsfuzzy \varphi\) iff \(w \modelsfuzzy \varphi\).
\end{lemma}

\begin{proof}
    We prove this \lcnamecref{preorderdep} by induction on the structure of the formula \(\varphi\).

    \textbf{Base case:} if \(\varphi \in \propatoms\) then \(v \modelsfuzzy \varphi\) iff \(v(\varphi) \geq \theta\). 
    And by \cref{def:inducedpreorder} \(v(\varphi) \geq \theta\) iff \(\varphi \preceq_v s_\theta\).
    As we assume \(\preceq_v = \preceq_w\), we have that \(\varphi \preceq_w s_\theta\).
    Using again \cref{def:inducedpreorder} and the definition of \(\modelsfuzzy\), we can conclude that \(\varphi \preceq_w s_\theta\) iff  \(v \modelsfuzzy \varphi\). 
    Therefore, if \(\varphi \in \propatoms\) then \(v \modelsfuzzy \varphi\) iff \(w \modelsfuzzy \varphi\).

    \textbf{Induction step:} Now, we assume that for all formulas \(\psi \in \llangfuzzy\) with length (number of connectives) at most \(n\), it holds that whenever \(\preceq_v = \preceq_w\) then \(v \modelsfuzzy \psi\) iff \(w \modelsfuzzy \psi\).
     We will consider now a formula \(\varphi \in \llangfuzzy\) that has length \(n + 1\), and treat each of the following cases separately.
    \begin{description}[leftmargin=*]
        \item[\(\varphi = \neg{\psi}\):] First, we remark that as a consequence \cref{def:fgodelss}, every valuation in \(\mUnifuzzy\) must assign 0 to \(a \land \neg{a}\) for \(a \in \propatoms\) (\(\propatoms \neq \emptyset\)).
            Consequently, every minimal element in the induced total preorder must be assigned 0 by the corresponding valuation.
            On the other hand, every formula assigned 0 by a valuation will be a minimal element in the induced total preorder.
            Thus, due to the semantics of negation in \(\logsysfuzzy\), \(v \modelsfuzzy \neg{\psi}\) iff \(\psi\) is a minimal element in \(\preceq_v\).
            By our assumption that \(\preceq_v = \preceq_w\), we can use the same argument to conclude that \(w \modelsfuzzy \neg{\psi}\) iff \(\psi\) is a minimal element in \(\preceq_w\).
            Hence, \(v \modelsfuzzy \neg{\psi}\) iff \(w \modelsfuzzy \neg{\psi}\).
        \item[\(\varphi = \psi_1 \land \psi_2\):] We know that \(v \modelsfuzzy \psi_1 \land \psi_2\) iff \(v(\psi_1) \geq \theta\) and \(v(\psi_2) \geq \theta\).
            In other words, \(v \modelsfuzzy \psi_1 \land \psi_2\) iff \(\psi_1 \preceq_v s_\theta\) and \(\psi_2 \preceq_v s_\theta\).
            Using the assumption that \(\preceq_v = \preceq_w\) and the induction hypothesis, we get that \(v \modelsfuzzy \psi_1 \land \psi_2\) iff \(w \modelsfuzzy \psi_1 \land \psi_2\).
        \item[\(\varphi = \psi_1 \lor \psi_2\):] We know that \(v \modelsfuzzy \psi_1 \lor \psi_2\) iff \(v(\psi_1) \geq \theta\) or \(v(\psi_2) \geq \theta\).
            In other words, \(v \modelsfuzzy \psi_1 \lor \psi_2\) iff \(s_\theta \preceq_v \psi_1 \) or \(s_\theta \preceq_v \psi_2\).
            Using the assumption that \(\preceq_v = \preceq_w\) and the induction hypothesis, we get that \(v \modelsfuzzy \psi_1 \lor \psi_2\) iff \(w \modelsfuzzy \psi_1 \lor \psi_2\).
        \item[\(\varphi = \psi_1 \rightarrow \psi_2\):] \(v \modelsfuzzy \psi_1 \rightarrow \psi_2\) iff (i) \(v(\psi_1) \leq v(\psi_2)\) or (ii) \(v(\psi_2) \geq \theta\).
            In other words, \(v \modelsfuzzy \psi_1 \rightarrow \psi_2\) iff \(\psi_1 \preceq_v \psi_2\) or \(s_\theta \preceq_v \psi_2\).
            As in the case of \(\land\) and \(\lor\), we can employ the assumption that \(\preceq_v = \preceq_w\) together with the induction hypothesis to conclude \(v \modelsfuzzy \psi_1 \rightarrow \psi_2\) iff \(w \modelsfuzzy \psi_1 \rightarrow \psi_2\).
    \end{description}
    Hence, if \(\preceq_v = \preceq_w\) then \(v \modelsfuzzy \varphi\) iff \(w \modelsfuzzy \varphi\).
\end{proof}

\begin{proposition}%
\label{possiblevalues}
    Let \(\theta \in \left(0,1\right]\), \(\propatoms\) be a non-empty finite set of propositional atoms, \(\logsysfuzzy = (\llangfuzzy, \mUnifuzzy, \modelsfuzzy)\) as in \cref{def:fgodelss}.
    Then, for any \(v \in \mUnifuzzy\) and \(\varphi \in \llangfuzzy\), \(v(\varphi) \in \{0, 1\} \cup \{v(a) \mid a \in \propatoms\}\).
\end{proposition}

\begin{proof}
    This clearly holds for \(\varphi \in \propatoms\).
    For complex formulas we just need to consider the possible valuations defined in the semantics of the connectives in \cref{def:fgodelss}.
    For all of the connectives, the valuation is either one of the values of the subformulas, 0 or 1.
\end{proof}

\begin{lemma}%
\label{prepre}
    Let \(v, w \in \mUnifuzzy\) with \(\preceq^\prime_v = \preceq^\prime_w\), then \(\preceq_v = \preceq_w\).
\end{lemma}

\begin{proof}
    From \cref{def:fgodelss}, the values assigned to all formulas in \(\llangfuzzy\) depend only on the valuations on \(\propatoms\).
    This means that there is only one possible way to extend a valuation on \(\propatoms\) to \(\llangfuzzy\).
    Moreover, it follows from \cref{possiblevalues} that every formula in \(\llangfuzzy\) can only assume values in \(\in \{0, 1\} \cup \{v(a) \mid a \in \propatoms\}\).
    Furthermore, as a consequence of the semantics of the connectives, for every \(\varphi \in \llangfuzzy\), if \(\varphi = \neg{\psi}\), then \(v(\varphi)\) depends on whether \(\psi\) is a minimal element in \(\preceq_v\), otherwise, if \(\varphi = \psi_1 \circ \psi_2\) with \(\circ \in \{\land, \lor, \rightarrow\}\) then \(v(\varphi)\) depends only on the relation between \(\psi_1\) and \(\psi_2\) according to \(\preceq_v\).
    As each formula will receive values from a finite set depending only on the total preorder induced on the propositional atoms, for any valuation \(v\), \(\preceq^\prime_v\) determines \(\preceq_v\).
\end{proof}

\begin{theorem}\label{fuzzyfinite}
    \(\FRsets(\logsysfuzzy)\) is finite.
\end{theorem}

\begin{proof}
    From \cref{preorderdep,prepre} and the definition of \(\modelsfuzzy\), we can conclude that whether \(v \modelsfuzzy B\), for \(v \in \mUnifuzzy\) and \(\baseb \in \finitepwset(\llangfuzzy)\), depends only on \(\preceq_v\).
    However, as the induced total preorders over \(\propatoms \cup \{s_\theta\}\) are defined over a finite set, there is a finite amount of distinct ones.
    In fact, there are at most \(\sum_{i = 0}^{|\propatoms| + 1} k!S(n, i)\) such preorders, where \(S(n, k)\) denotes the Stirling partition number.
    This implies that while there infinitely many valuations in \(\mUnifuzzy\), there is only a finite number subsets of \(\mUnifuzzy\) that can be represented via a base in \(\llangfuzzy\).
    Therefore, \(\FRsets(\logsysfuzzy)\) must be finite.
\end{proof}

\rLFuzzyCompat*

\begin{proof}
    It follows from \cref{evcRcpSuff} and \cref{fuzzyfinite} that we only need to prove that \(\emptyset, \mUnifuzzy \in \FRsets(\logsysfuzzy)\).
    Let \(\baseb_\bot = \{\neg{a} \land a\}\) for some \(a \in \propatoms\) and also let \(\baseb_\top = \emptyset\).
    As a consequence of \cref{def:fgodelss}, for any \(\theta \in \left(0, 1\right]\) and valuation \(v \in \mUnifuzzy\): \(v(\baseb_\bot) = 0 < \theta\) and \(v(\baseb_\top) = 1 \geq \theta\).
    Therefore, \(\logsysfuzzy\) is \mconnm-compatible and \mexnm-compatible.
\end{proof}

\subsection{Proofs for \Cref{use:ltlnext}}


\begin{restatable}{proposition}{rallfinite}\label{obs:allfinite}
    In $\logsysltlx$,  every finite set of formulae is a theory, that is, for every \(\baseb \in \finitepwset(\llangltlx)\) and \(\varphi \in \llangltlx\), if \(M \in \modelsof{\baseb}\), then \(M \models \{\varphi\}\) iff \(\varphi \in \baseb\).  
\end{restatable}

\begin{proof}
    It suffices to show that for every finite set $\baseb$,  and formula $\varphi \not \in \baseb$,  
    there is a model $(M,s)$ such that $(M,s) \models \baseb$ but $(M,s) \not \models \{\varphi\}$. 
    By definition,  every formula in this logics is of the form $X^n p$,  where $p$ is an atomic propositional formula and $n \in \mathbb{N}$.  
    Let $m = \max(\{ k \in \mathbb{N} \mid X^k p \in  \baseb \cup \{\varphi\}\})$.  
    The value $m$ contains the highest value of $X^k$ of the formulae in $\baseb \cup \{\varphi\}$. 
    This works as an upper bound on the size of the model $M$,  we will construct.  
    Let us construct the model $M = (S,  R,  \lambda)$ where
    \begin{itemize}
        \item $S = \{s_1, \dots,  s_m\}, $
        \item $R = \{  (s_i,  s_{i+1}) \mid  i < m  \} \cup \{(s_m, s_m)\}$ 
        \item $\lambda(s_i) = \{  p \in \propatoms \mid X^i p \in \baseb \}$ 
    \end{itemize}
    Observe that $M$ is indeed a Kripke structure.   
    We only need to show that (1) $(M,s_1) \modelsltlx \baseb$ and (2) $(M,s_1) \not \modelsltlx \{\varphi\}$. 
    \begin{enumerate}[label= (\arabic*), leftmargin=*]
        \item let $\psi \in \baseb$.  Thus,  $\psi = X^i p$,  for some $i \geq 0$.  
            By definition of $M$,  $p  \in \lambda(s_i)$, which means that,   $(M,s_1) \modelsltlx X^i p$. 	
            Therefore,  $(M,s_1) \modelsltlx \psi$,  for all $\psi \in \baseb$,  that is,  $(M,s_1) \modelsltlx \baseb$.  

        \item   $\varphi = X^i q$,  for some $i \geq 0$.  By hypothesis,   $\varphi \not \in X$.  Thus,  by definition of $\lambda$,  we get  $q \not \in \lambda(i)$.  
            Thus,  $(M,  s_1) \not \modelsltlx X^i q$,   that is,   $(M,  s_1) \not \modelsltlx \varphi$.
    \end{enumerate}
\end{proof}

\rltlxrecep*

\begin{proof}

    Let us suppose for contradiction that $\rcpx(\baseb, \mSet) \not \in \FRsups(\baseb \cup  \mSet, \logsysltlx)$.  
    Thus,  there is some finite representable $Y \subseteq \mUni$ such that 
    \begin{align}
        \modelsof{\baseb} \cup \mSet \subseteq Y \subset \modelsof{\rcpx(\baseb, \mSet)},  \label{eq:modelsubXY}
    \end{align}
    and $Y = \modelsof{\baseb_Y}$,  for some finite $\baseb_Y$.  
    From \cref{obs:allfinite}, every finite set is a theory,   which implies that  both $\rcpx(\baseb, \mSet)$ and $\baseb_Y$ are finite theories.  
    Thus,  as the logic is monotonic,  we get that $\rcpx(\baseb, \mSet) \subset \baseb_Y$. 
    Thus there is some $\varphi \in \baseb_Y$ such that $\varphi \not \in \rcpx(\baseb, \mSet)$.  
    By definition, 
    \begin{align*}
        \psi \in \rcpx(\baseb, \mSet) &\mbox{ iff } \modelsof{\baseb} \cup \mSet \modelsltlx \psi
    \end{align*}  
    Thus,  as $\varphi \not \in \rcpx(\baseb, \mSet)$,  we get that $\modelsof{\baseb} \cup \mSet \not \modelsltlx \varphi$.  
    However,   from~\eqref{eq:modelsubXY},  we have that  $\modelsof{\baseb} \cup \mSet \subseteq Y$.  
    Thus,  as $\varphi \in \baseb_Y$,  
    we have that $\modelsof{\baseb} \cup \mSet  \modelsltlx \varphi$,  which is a contradiction.   
\end{proof}

\begin{proposition}\label{ltlxrcompat}
    $\logsysltlx$ is \mexnm-compatible.  
\end{proposition}
\begin{proof}
    It follows from \cref{prop:ltlxrecep} that
    $ \rcpx(\baseb, \mSet) \in \FRsups(\baseb \cup \mSet, \logsysltlx)$,  which means that 
    $\rcpx$ is a maxichoice reception function. Therefore,  according to \cref{mFRExp_r}, $\rcpx$ satisfies all rationality postulates for reception.  Thus, $\logsysltlx$ is reception-compatible.  
\end{proof}

\begin{proposition}\label{ltlxecompat}
    $\logsysltlx$ is not \mconnm-compatible.  
\end{proposition}
\begin{proof}
    Let $\baseb = \{p\}$ and \(\mSet = \modelsof{\baseb}\).
    Thus,  eviction of $\baseb$ by $\mSet$ must result in a finite base $\baseb'$ such that $\modelsof{\baseb'} = \emptyset$.
    However, the empty set of models is not finitely representable in this logic.
    To prove this, it is enough to show that every finite base \(\baseb \in \finitepwset(\logsysltlx)\) has a model.
    In fact, for any \(\baseb \in \finitepwset(\llang)\), the model \(M = (S, R, \lambda)\) such that
    \(S = \{s\}\), \(R  = \{(s, s)\}\) and \(\lambda(s) = \{p\}\) will satisfy any finite set of formulae of the form \(X^{k} p\) with \(k \in \mathbb{N}\). 
\end{proof}

\rltlxcompats*

\begin{proof}
    Follows directly from \cref{ltlxrcompat,ltlxecompat}.
\end{proof}

\subsection{Proofs for \cref{use:alc}}

In the following proofs, we will consider in this work standard abbreviations
for concept constructors in \ALC{} that were not describe in \cref{use:alc}. 
For example  $\bot$ is interpreted as the empty set 
and $\top$
is interpreted as the whole domain.
In some of the proofs, we will also employ the fact that
usual concept inclusions $C \sqsubseteq D$ can be expressed equivalently
as  $\top \sqsubseteq \neg C\sqcup D$ and $\neg C\sqcup D \sqsubseteq \top$.
We will also write \(\exists r^m.C\) to denote the nesting of the existential
restriction \(\exists r\) \(m\) times over the concept $C$.
We establish in Theorem~\ref{ALCrcompat} that
     \(\logsysALC\) is not 
 \mconnm-compatible. Our proof
 holds both in the case in which the disjoint sets
 $\NC,\NR,\NI$ are assumed to be finite
 or (countably) infinite.
\begin{restatable}{theorem}{rcpALCcompat}%
\label{ALCrcompat}
    \(\logsysALC\) is not 
    \mconnm-compatible. 
\end{restatable}

\begin{proof}
    Let \(\logsysALC = (\llang_\ALC, \mUni_\ALC, \models_\ALC)\) be the
usual satisfaction system for \(\ALC\). 
For conciseness, we will write \(\models\) instead of \(\models_\ALC\) within this proof.  
Let \(\baseb_\top =\{\bot\sqsubseteq \top\}\), that is, \(\modelsof{\baseb_\top} = \mUni\). 
Also, given a fixed but arbitrary $a\in\NI$ and $r\in\NR$, we define models of the form
\(M^n=(\mathbb{N},\cdot^{M^n})\)
where 
\[r^{M^n} = \{(i,i+1)\mid i\in \mathbb{N}, 0\leq i < n\}\]
and $a^{M_n}=0$, and similarly
\(M^{\infty}=(\mathbb{N},\cdot^{M^{\infty}})\)
where \[r^{M^{\infty}} = \{(i,i+1)\mid i\in \mathbb{N}\}\]
and $a^{M^{\infty}}=0$.
Let $\mSet$ be the set of all models $M$
such that for some $n\in\mathbb{N}$
we have that $a^M\in (\forall r^n.\bot)^M$.
That is, there is no loop or infinite chain
of elements connected via the role $r$ starting
from $a^M$.
By definition of $\mSet$, we have that
$M^{\infty}\not\in \mSet$ since
this model has an
infinite chain
of elements connected via the role $r$ starting
from $a^M$, while $M^{n}\in \mSet$ for all $n\in\mathbb{N}$.

 To prove that \(\logsysALC\) is not \mconnm-compatible, we need to prove that there is no \(\baseb \in \finitepwset(\llang_\ALC)\) such that \(\modelsof{\baseb} \in \FRsubs(\mSet, \logsysALC)\), that is, \(\FRsubs(\mSet, \logsysALC) = \emptyset\).
 Intuitively, we want to show that 
 we cannot find a maximal \ALC{} ontology 
 that finitely represents the result of removing the models in $\mUni \setminus \mSet$ 
 from $\baseb_\top$.
 First, we  show 
 the following claims. 
\begin{claim}\label{cl:eviction-aux-aux}
	For every \ALC{} concept $C$ 
	if $M^{\infty}\models C(a)$
	then there is $n\in\mathbb{N}$ 
	such that for all $m\geq n$, with $m\in\mathbb{N}$, 
	we have that $M^m\models C(a)$.
\end{claim}

The proof is by structural induction.
We assume w.l.o.g. that $C$
is in negation normal form,
which means that we need to deal with 
expressions of the form
$\exists r.D_1$, $\forall r.D_1$, $D_1\sqcap D_2$,
$D_1\sqcup D_2$ (but we can disregard $\neg D_1$). 
In the base case we have $C=\exists r.\top$
and $C=\forall r.\top$.
The claim holds in the base case 
since, by definition of $M^{n}$, we have that  
$M^{n}\models \exists r.\top(a)$,
 for all $n\in\mathbb{N}$,
and the premisse is violated for $\forall r.\bot$ (that is, $M^{\infty}\not\models \forall r.\bot(a)$).
Suppose that the claim holds for $D\in\{D_1,D_2\}$,
that is,	if $M^{\infty}\models D(a)$ then there is $n\in\mathbb{N}$
such that for all $m\geq n$, we have that 
$M^m\models D(a)$.
We now consider the following cases.
\begin{itemize}
	\item  $\exists r.D_1$: 
	Suppose that $M^{\infty}\models \exists r.D_1(a)$.
	By definition of $M^{\infty}$, we have that 
	$M^{\infty}\models D_1(a)$ and so,
	by the inductive hypothesis, 
	there is  $n\in\mathbb{N}$
	such that for all $m\geq n$, we have that
	$M^m\models D_1(a)$. 
	By definition of $M^m$, for all $m\geq n$, with $m\in\mathbb{N}$,
	we have that $M^{m+1}\models \exists r.D_1(a)$. 
	\item  $\forall r.D_1$: 
		Suppose that $M^{\infty}\models \forall r.D_1(a)$.
	By definition of $M^{\infty}$, we have that 
	$M^{\infty}\models D_1(a)$ and so,
	by the inductive hypothesis, 
	there is  $n\in\mathbb{N}$
	such that for all $m\geq n$, we have that
	$M^m\models D_1(a)$. 
	By definition of $M^m$, for all $m\geq n$, with $m\in\mathbb{N}$,
	we have that $M^{m+1}\models \forall r.D_1(a)$. 
	\item $D_1\sqcap D_2$: 	
	Suppose that $M^{\infty}\models D_1\sqcap D_2(a)$.
	Then, $M^{\infty}\models D_1(a)$
	and $M^{\infty}\models D_2(a)$.
	By the inductive hypothesis,
	there are $n_1,n_2$ such that 
	for all $m_1\geq n_1$ and all $m_2\geq n_2$, we have that
	$M^{m_1}\models D_1(a)$ and $M^{m_2}\models D_2(a)$.
	Assume w.l.o.g. that $n_1\geq n_2$.
	Then, 	
	for all $m\geq n_1$, we have that
	$M^m\models D_1\sqcap D_2(a)$. 
	\item $D_1\sqcup D_2$:
		Suppose that $M^{\infty}\models D_1\sqcup D_2(a)$.
	Then, $M^{\infty}\models D_1(a)$
	or $M^{\infty}\models D_2(a)$.
	Assume w.l.o.g. that $M^{\infty}\models D_1(a)$.
	By the inductive hypothesis,
	there is $n$ such that 
	for all $m\geq n$, we have that
	$M^{m}\models D_1(a)$.
	Then, 	
	for all $m\geq n$, we have that
	$M^m\models D_1\sqcup D_2(a)$. 
\end{itemize}

\begin{claim}\label{cl:eviction-aux}
	For every \ALC{} concept $C$ 
	if there is $n\in\mathbb{N}$ 
	such that for all $m\geq n$, with $m\in\mathbb{N}$, 
	we have that $M^m\models C(a)$
	then $M^{\infty}\models C(a)$.
\end{claim}
The proof is by structural induction but we 
do not use negation normal form in this proof.
In the base case we have $C=\exists r.\top$. 
The claim holds in the base case  for all $n\in\mathbb{N}$
and all $m\geq n$, with $m\in\mathbb{N}$,
since, by definition of $M^{\infty}$, we have that  
$M^{\infty}\models \exists r.\top(a)$.
Suppose that the claim holds for $D\in\{D_1,D_2\}$,
that is,	if there is $n\in\mathbb{N}$
such that for all $m\geq n$, we have that 
 $M^m\models D(a)$ 
then $M^{\infty}\models D(a)$.
We now consider the following cases.
\begin{itemize}
\item  $\exists r.D_1$: 
Suppose that there is  $n\in\mathbb{N}$
such that for all $m\geq n$, we have that
  $M^m\models \exists r.D_1(a)$. 
  By definition of $M^m$, for all $m\geq n >0$,
   we have that $M^{m-1}\models D_1(a)$ (note that we 
   can assume w.l.o.g. that there is such $n>0$ because
    if there is $n$ satisfying the claim then $n+1$ also satisfies the claim). 
  Then, by the inductive hypothesis,
  $M^{\infty}\models D_1(a)$.
  Finally, by definition of 
   $M^{\infty}$, if   $M^{\infty}\models D_1(a)$ then  $M^{\infty}\models \exists r.D_1(a)$.
\item $\neg D_1$: In this case, we use the contrapositive.
Suppose that $M^{\infty}\not\models \neg D_1(a)$.
We want to show that there is no 
$n\in\mathbb{N}$ 
such that for all $m\geq n$, with $m\in\mathbb{N}$, 
we have that $M^m\models \neg D_1(a)$.
If $M^{\infty}\not\models \neg D_1(a)$
then  $M^{\infty} \models  D_1(a)$ and so,
 by Claim~\ref{cl:eviction-aux-aux},
 there is  $n\in\mathbb{N}$
such that for all $m\geq n$, with $m\in\mathbb{N}$, 
we have that
$M^m\models    D_1(a)$. 
Then, there can be no 
$n \in\mathbb{N}$
such that for all $m\geq n$, with $m\in\mathbb{N}$, 
we have that
$M^{m}\models \neg   D_1(a)$.
\item $D_1\sqcap D_2$:
Suppose that there is  $n\in\mathbb{N}$
such that for all $m\geq n$, we have that
$M^m\models D_1\sqcap D_2(a)$. 
Then, for all $m\geq n$, we have that
$M^m\models D_1(a)$ and $M^m\models D_2(a)$.
By the inductive hypothesis,
$M^{\infty}\models D_1(a)$ and $M^{\infty}\models D_2(a)$.
So $M^{\infty}\models D_1\sqcap D_2(a)$.
\end{itemize}

\begin{claim}\label{cl:eviction-aux-aux-aux}
	For every \ALC{} concept $C$ 
	if there is $n\in\mathbb{N}$ 
	such that for all $m\geq n$, with $m\in\mathbb{N}$, 
	we have that $M^m\models \top\sqsubseteq C$
	then $M^{\infty}\models \top\sqsubseteq C$.
\end{claim}

Suppose to the contrary that, for some \ALC{} concept $C$, 
    there is $n\in\mathbb{N}$ 
such that, for all $m\geq n$, with $m\in\mathbb{N}$, 
we have that $M^m\models \top\sqsubseteq C$
but $M^{\infty}\not\models \top\sqsubseteq C$.
If $M^{\infty}\not\models \top\sqsubseteq C$
then there is $k\in\mathbb{N}$ such that
$k\not\in C^{M^{\infty}}$. 
By definition of $k\in\mathbb{N}$ 
and the models of the form $M^n$ (recall that
the domain of such models is $\mathbb{N}$),
for all $m\in\mathbb{N}$, there is a bisimulation 
between $M^{m-k}$ and $M^{m}$ containing 
$(a^{M^{m-k}},k)$.
Since \ALC{} is invariant under bisimulations, 
for all $m'\geq m-k$,
we have that $a^{M^{m'}} \in C^{M^{m'}}$. 
Then, by Claim~\ref{cl:eviction-aux}, 
$a^{M^{\infty}} \in C^{M^{\infty}}$. 
By definition of $k\in\mathbb{N}$ 
and  $M^{\infty}$,
there is a bisimulation between $M^{\infty}$ and 
itself (that is, $M^{\infty}$) containing 
$(a^{M^{\infty}},k)$.
Therefore, $k \in C^{M^{\infty}}$. 

\medskip

We are now ready to show that 
\(\logsysALC\) is not \mconnm-compatible.
Suppose to the contrary that
there is \(\baseb \in \finitepwset(\llang_\ALC)\) such that \(\modelsof{\baseb} \in \FRsubs( \mSet, \logsysALC)\).
We can assume w.l.o.g. that $\baseb$ is of the form
$\{\top\sqsubseteq D, C(a)\}$.
Indeed, if it contains
e.g. $C_1(a), \ldots, C_k(a)$ then this is equivalent 
to $C_1\sqcap\ldots \sqcap C_k(a)$.
Also, concept inclusions $C_1\sqsubseteq D_1, \ldots, C_k\sqsubseteq D_k$ can be equivalently
rewritten as $\top\sqsubseteq ((\neg C_1 \sqcup D_1)\sqcap \ldots \sqcap (\neg C_k \sqcup D_k))$.

If there is $n\in\mathbb{N}$ such that
$M^n\not\models \baseb$ then\footnote{Recall that
	$M^n$ has a chain of $n+1$ elements connected via the role $r$.}
\[\baseb':=\{\top\sqsubseteq D\sqcup (\bigsqcup^{n+1}_{i=0}(\exists r^{i}.\top\sqcap \neg \exists r^{i+1}.\top)), 
\]\[
C\sqcup (\exists r^{n+1}.\top\sqcap \neg \exists r^{n+2}.\top)(a)\}\]
is such that $M^n\models \baseb'$.
Moreover, $\modelsof{\baseb}\subset\modelsof{\baseb'}$.
By definition of $\baseb'$ and $\mSet$, we also have that
  \(\modelsof{\baseb'} \in \FRsubs( \mSet, \logsysALC)\).
This contradicts the assumption that
\(\modelsof{\baseb} \in \FRsubs( \mSet, \logsysALC)\).
So, for all $n\in\mathbb{N}$, we have  that
$M^n\models \baseb$.
 
Then, by Claims~\ref{cl:eviction-aux} and~\ref{cl:eviction-aux-aux-aux}, it follows that  
$M^{\infty}\models \baseb$.
Since, as already mentioned, $M^{\infty}\not\in\mSet$,
this contradicts the assumption that
\(\modelsof{\baseb} \in \FRsubs( \mSet, \logsysALC)\).
  Thus, $\FRsubs( \mSet, \logsysALC)=\emptyset$. 
\end{proof}

We now prove Theorem~\ref{ALCecompat}
(the signature is infinite).

\begin{restatable}{theorem}{evcALCcompat}%
\label{ALCecompat}
    \(\logsysALC\) is not 
    \mexnm-compatible. 
\end{restatable}
	
	
\begin{proof}
    Let \(\logsysALC = (\llang_\ALC, \mUni_\ALC, \models_\ALC)\) be the
    usual satisfaction system for \(\ALC\). 
    For conciseness, we will write \(\models\) instead of \(\models_\ALC\) within this proof.  
    Assume for contradiction that
    \(\logsysALC\) is \mexnm-compatible.

    Consider the signature \(\NC = \{C_i \mid i\in \mathbb{N}\}\), \(\NI = \{a_i \mid i \in \mathbb{N}\}\), and \(\NR\) an arbitrary countably infinite set disjoint with \(\NC \cup \NR\).
    Also, consider the model \(\modelm
    = (\Delta^\modelm, \cdot^\modelm)\}\) where \(\Delta^\modelm = \mathbb{N}\), and
    \(\cdot^\modelm\) is such that \(r^\modelm = \emptyset\) for all \(r \in \NR\) and \(A_i^\modelm = \{a_i\}\) and \(a_i^\modelm = i\) for all \(i \in \mathbb{N}\). 
    Now, let \(\baseb_\bot = \{\top
    \sqsubseteq \bot\}\), \(\baseb_\bot\) is inconsistent (it has no models).

    By hypothesis,  \(\logsysALC\) is \mexnm-compatible which means that 
    there is a \(\baseb \in
    \finitepwset(\llang_\ALC)\) such that \(\modelsof{\baseb}
    = \modelsof{\mExp(\baseb_\bot, \{M\})}\), that is, \(\modelsof{\baseb} \in \FRsups(\{M\}, \logsysALC)\). 
Let 
    \begin{align*}
        J = \{i \in \mathbb{N} \mid \forall M', M^{''} \in \modelsof{\baseb},\\
		 M' \models A_i(a_i) \text{ iff } M^{''} \models A_i(a_i)\}
    \end{align*}
    We have two cases: either (i) \(J \neq \mathbb{N}\), or (ii) \(J = \mathbb{N}\).
    In all cases,  we will reach  a contradiction,  and therefore we conclude that \(\logsysALC\) is not \mexnm-compatible. 

\begin{enumerate}[label= (\roman*), leftmargin=*]
    \item \(J \neq \mathbb{N}\). Then \(\baseb\) does not specify whether some \(A_k(a_k)\) with \(k \in \mathbb{N} \setminus J\) holds or not, that is, it will have both models in which \(A_k(a_k)\) holds and models in which \(\neg{A_k}(a_k)\) holds.
    We can build a base \(\baseb' = \baseb \cup \{A_k(a_k)\}\). 
    The base \(\baseb'\) is finite, \(M \in \modelsof{\baseb'}\), and \(\modelsof{\baseb'} \subset \modelsof{\baseb}\). 
    Hence, \(\modelsof{\baseb} \not\in \FRsups(\{M\}, \logsysALC)\), a contradiction.

    \item \(J = \mathbb{N}\). In this case, \(M \models A_i(a_i)\) for all \(M \in \modelsof{\baseb}\) and all $i\in\mathbb{N}$.    
        Without loss of generality, we can assume that \(\baseb\) has a single concept inclusion \(\top \sqsubseteq P\), with \(P\) an \ALC{} concept.
        Moreover, as \(\baseb\) is finite, it can only have finitely many assertions.
        Thus, we write \(\baseb\) as 
        \[
            \baseb = \{\top \sqsubseteq P\} \cup \{C_k(a_k) \mid k \in K\},
        \]
        where \(K\) is a finite subset of \(\mathbb{N}\) and \(P\) and \(C_k\) are \ALC{} concepts for all \(k \in K\).
        Let \(j  \in \mathbb{N} \setminus K\).
        As the assertions cannot enforce \(A_j(a_j)\), we have, by the monotonicity of \ALC{}, that \( \models \top\sqsubseteq P\) must entail  \(A_j(a_j)\), in other words, \(P \sqsubseteq A_j\) must be a tautology.
        But this also implies that \(\top \sqsubseteq A_j\) must hold in every model of \(\{\top \sqsubseteq P\}\).
        However, \(M \not\models A_j(a_i)\) for all \(i \neq j\), therefore \(M\) is not a model of \(\{\top \sqsubseteq P\}\) (\(M \not\in \modelsof{\{\top \sqsubseteq P\}}\)). 
        Additionally, by semantics and monotonicity of \ALC{}, we have that \(\modelsof{\baseb} \subseteq \modelsof{\{\top \sqsubseteq P\}}\), thus \(M \not\in \modelsof{\baseb}\), a contradiction. 
\end{enumerate}

    Therefore, there is no finite base \(\baseb\) such that: \(\modelsof{\baseb}
    \in \FRsets(\logsysALC)\), \(M \in \modelsof{\baseb}\) and
    \(\modelsof{\baseb}\) is \(\subseteq\)-minimal. Consequently, \(\FRsups(\{M\},
    \logsysALC) = \emptyset\).  
    Hence, \(\logsysALC\) is not \mexnm-compatible.
\end{proof}

\allALCcompat*

\begin{proof}
    Direct consequence of \cref{ALCrcompat,ALCecompat}
\end{proof}

We now consider a simpler satisfaction system that we called \(\logsysDLABox\).  We point out that 
we need negative assertions
to be \mconnm-compatible since we cannot express
contradiction with only positive assertions 
(logics that cannot express contradiction
are not \mconnm-compatible).

\allABoxcompat*
\begin{proof}
	The proof that  \(\logsysDLABox\) is  not 
	\mexnm-compatible
	is similar to the proof of Theorem~\ref{ALCecompat}, but simpler since 
	we only need to consider assertions (not concept inclusions). 	
	
 We now show that  \(\logsysDLABox\) is \mconnm-compatible.
 For this we need to prove that, for every
set of (positive and negative) assertions 
\Omc---we call it an \empty{ABox ontology}---and every \(\mSet \subseteq \mUni\),
 we have that  \(\FRsubs(\modelsof{\Omc} \setminus \mSet, \logsysDLABox) \neq \emptyset\).

Suppose to the contrary that
there exists an ABox ontology $\hat{\Omc}$
and a set \(\mSet \subseteq \mUni\)
such that 
\(\FRsubs(\modelsof{\hat{\Omc}} \setminus \mSet, \logsysDLABox) = \emptyset\).
By definition of \(\FRsubs(\modelsof{\hat{\Omc}} \setminus \mSet, \logsysDLABox)\),
this can only happen
if either
\begin{itemize}
\item there is no ABox ontology \Omc such that 
$\modelsof{\Omc}\subseteq (\modelsof{\hat{\Omc}} \setminus \mSet)$, or
\item 
for all $i\in\mathbb{N}$,
there are ABox ontologies
$\Omc_i,\Omc_{i+1}$
such that 
$\modelsof{\Omc_i}\subset\modelsof{\Omc_{i+1}}\subseteq (\modelsof{\hat{\Omc}} \setminus \mSet)$.

\end{itemize}
The former cannot happen
since we can express contradiction in 
our restricted language (so there is always
an ABox ontology \Omc, e.g. $A(a),\neg A(a)$,   such that $\modelsof{\Omc}=\emptyset$
and this is for sure a subset of 
$(\modelsof{\hat{\Omc}} \setminus \mSet)$).
The fact that the latter also cannot
happen is because given two ontologies
$\Omc,\Omc'$ in this restricted
language we have that
$\modelsof{\Omc}\subset
\modelsof{\Omc'}$
iff $\Omc'\subset
 \Omc$.
So
there cannot 
be an infinite sequence 
of ABox ontologies $\Omc_i$ 
such that, for all $i\in\mathbb{N}$, 
$\modelsof{\Omc_i}\subset
\modelsof{\Omc_{i+1}}$ because this implies
$\Omc_i\supset \Omc_{i+1}$, for all $i\in\mathbb{N}$, 
and $\Omc_i, \Omc_{i+1}$ are finite.
\end{proof}
\allDLLitecompat*
\begin{proof}
	We start proving
	that \(\logsysDLLITE\)
	is eviction-compatible. 
    For conciseness, we will write \(\mUni\) and \(\models\) to represent, respectively, the universe of models and the satisfaction relation in \(\logsysDLLITE\) in this proof.
For this we need to prove that, for every
DL-Lite$_\Rmc$ ontology \Omc
and every \(\mSet \subseteq \mUni\),
we have that  \(\FRsubs(\modelsof{\Omc} \setminus \mSet, \logsysDLLITE) \neq \emptyset\).

Suppose to the contrary that
there exists a DL-Lite$_\Rmc$ ontology $\hat{\Omc}$
and a set \(\mSet \subseteq \mUni\)
such that 
\(\FRsubs(\modelsof{\hat{\Omc}} \setminus \mSet, \logsysDLLITE) = \emptyset\).
By definition of \(\FRsubs(\modelsof{\hat{\Omc}} \setminus \mSet, \logsysDLLITE)\),
this can only happen
if either
\begin{itemize}
\item there is no DL-Lite$_\Rmc$ ontology \Omc such that 
$\modelsof{\Omc}\subseteq (\modelsof{\hat{\Omc}} \setminus \mSet)$, or
	\item 
for all $i\in\mathbb{N}$,
there are DL-Lite$_\Rmc$ ontologies
$\Omc_i,\Omc_{i+1}$
such that 
$\modelsof{\Omc_i}\subset\modelsof{\Omc_{i+1}}\subseteq (\modelsof{\hat{\Omc}} \setminus \mSet)$.

\end{itemize}
The former cannot happen
since we can express contradiction in DL-Lite$_\Rmc$ (so there is always
a DL-Lite$_\Rmc$ ontology \Omc such that $\modelsof{\Omc}=\emptyset$
and this is for sure a subset of 
$(\modelsof{\hat{\Omc}} \setminus \mSet)$).
The fact that the latter also cannot
happen is a consequence of the
following two claims. 


\begin{claim}\label{claim:dllite}
  Given a satisfiable DL-Lite$_\Rmc$ ontology \Omc (over a finite signature $\NC\cup\NR\cup\NI$), we have that the DL-Lite$_\Rmc$ ontology
  $\Omc^t=\{\alpha\mid \Omc\models \alpha \}$
   is finite.
\end{claim}	

We first argue that the number 
of possible DL-Lite$_\Rmc$ concept and role inclusions that can be formulated
with a finite signature $\NC\cup\NR\cup\NI$
is finite. 
Indeed, concepts are of the form
$A,\neg A,\exists r,\neg  \exists r,\exists r^-,\neg  \exists r^-$ and role expressions
are of the form $R,R^-,\neg R, \neg R^-$. So if the number of concept names plus the number of role names is $n$, then there are at most
$(6n)^2$   possible concept inclusions and at most $(4n)^2$   possible role inclusions (with concept and role 
names occurring in  \Omc).
This finishes the proof of the claim.

\medskip



\begin{claim}\label{claim:dllite-aux}
	Let $\Omc,\Omc'$ be a satisfiable DL-Lite$_\Rmc$ ontologies.
If $\modelsof{\Omc}\subset\modelsof{\Omc'} $ then $\Omc^t\supset \Omc'^{t}$.
\end{claim}	

If $\modelsof{\Omc}\subset\modelsof{\Omc'} $ then 
$\Omc\models \Omc'$.
This means that
if $\Omc'\models \alpha$
then $\Omc\models \alpha$. So if
$\alpha$ is in $\Omc'^{t}$ 
then it is in $\Omc^{t}$.
In other words, $\Omc^t\supset \Omc'^{t}$.
This finishes the proof of the claim.

\medskip

By Claims~\ref{claim:dllite} and~\ref{claim:dllite-aux}  there cannot 
be an infinite sequence 
of DL-Lite$_\Rmc$ ontologies $\Omc_i$ 
such that, for all $i\in\mathbb{N}$, 
$\modelsof{\Omc_i}\subset
\modelsof{\Omc_{i+1}}$ because this implies  
$\Omc^t_i\supset \Omc^{t}_{i+1}$, for all $i\in\mathbb{N}$, 
and $\Omc^t_i,\Omc^{t}_{i+1}$ are finite.
\medskip

The proof 
that \(\logsysDLLITE\)
is reception-compatible is similar. 
For this we need to prove that, for every
DL-Lite$_\Rmc$ ontology \Omc
and every \(\mSet \subseteq \mUni\),
we have that
\(\FRsups(\modelsof{\Omc} \cup \mSet, \logsysDLLITE) \neq \emptyset\).


Suppose to the contrary that
there exists a DL-Lite$_\Rmc$ ontology $\hat{\Omc}$
and a set \(\mSet \subseteq \mUni\)
such that 
\(\FRsups(\modelsof{\hat{\Omc}} \cup \mSet, \logsysDLLITE) = \emptyset\).
By definition of \(\FRsups(\modelsof{\hat{\Omc}} \cup \mSet, \logsysDLLITE)\),
this can only happen
if either
\begin{itemize}
\item there is no DL-Lite$_\Rmc$ ontology \Omc such that 
$(\modelsof{\hat{\Omc}} \cup \mSet)\subseteq \modelsof{\Omc}$, or
\item 
for all $i\in\mathbb{N}$,
there are DL-Lite$_\Rmc$ ontologies
$\Omc_i,\Omc_{i+1}$
such that 
$(\modelsof{\hat{\Omc}} \cup \mSet)\subseteq \modelsof{\Omc_{i+1}} \subset
\modelsof{\Omc_i} $.
\end{itemize}
The former cannot happen
since we can express tautologies in DL-Lite$_\Rmc$ 
(so there is always
a DL-Lite$_\Rmc$ ontology \Omc such that $\modelsof{\Omc}=\mUni$
and this is for sure a superset of 
$(\modelsof{\hat{\Omc}} \cup \mSet)$).
The latter is  a consequence of   Claims~\ref{claim:dllite} and~\ref{claim:dllite-aux}.  There cannot 
be an infinite sequence 
of DL-Lite$_\Rmc$ ontologies $\Omc_i$ 
such that, for all $i\in\mathbb{N}$, 
$\modelsof{\Omc_{i+1}}\subset
\modelsof{\Omc_i}$ because this implies 
$\Omc^t_i\subset \Omc^{t}_{i+1}$, for all $i\in\mathbb{N}$, 
and $\Omc^t_i, \Omc^{t}_{i+1}$ are bounded by a polynomial in the
size of the finite signature (see proof of Claim~\ref{claim:dllite}).
\end{proof}

\end{document}